\documentclass[12pt,a4paper]{article}

\usepackage[centertags]{amsmath}
\usepackage{amsfonts,amsthm,amssymb}
\usepackage{graphicx}
\usepackage{bbm}
\usepackage{dsfont}
\usepackage{thmtools}
\usepackage{thm-restate}
\usepackage{mathtools}
\DeclarePairedDelimiter{\abs}{\lvert}{\rvert}
\usepackage{booktabs}
\usepackage{appendix}
\usepackage[hmargin=2.6cm,vmargin=2.6cm]{geometry}
\usepackage{etoolbox}
\usepackage{tikz}
\usepackage{comment}
\usepackage{float}
\usepackage{natbib}
\usepackage[hidelinks]{hyperref}
\usepackage{soul}
\usepackage{caption,subcaption}
\usepackage[shortlabels]{enumitem}
\usepackage{xcolor,framed}

\colorlet{lightgray}{gray!60}
\theoremstyle{definition}
\def\indic{\mathbbm{1}}

\newcommand{\real}{\mathbb{R}}

\newcommand{\be}{\begin{equation}}
\newcommand{\ee}{\end{equation}}

\newcommand{\1}{\mathbf{1}}
\newcommand{\defeq}{\mathrel{\mathop:}=}
\newcommand{\R}{\mathbf{R}}
\newcommand{\A}{\mathbf{A}}
\newtheorem{lemma}{Lemma}

\newtheorem{proposition}{Proposition}
\newtheorem{observation}{Observation}
\newtheorem{theorem}{Theorem}

\newtheorem{corollary}{Corollary}
\newtheorem{definition}{Definition}
\newtheorem{example}{Example}
\newtheorem{remark}{Remark}[section]
\title{Strategic Centrality and the Emergence of Core-Periphery Networks}

\author{
Itai Arieli\thanks{Department of Economics, University of Toronto and the Technion.}
\and
Jo\~ao Correia-da-Silva\thanks{FEP School of Economics and Management, University of Porto.}
\and
Wade Hann-Caruthers\thanks{FEP School of Economics and Management, University of Porto.}
\and
Anna Rubinchik\thanks{FEP School of Economics and Management, University of Porto.}
}

\begin{document}

\maketitle
\begin{abstract}
We study a network formation game in which agents sponsor links at a linear cost in order to maximize centrality, defined as a weighted sum of walk counts with positive and weakly decreasing weights. This class includes Katz Bonacich centrality and total communicability and captures environments in which access decays with distance. Our main result is a sharp equilibrium characterization: every Nash equilibrium network is core periphery, meaning there exists a set $C$ such that every node is linked to every node in $C$, and there are no other edges. The driving force is that linking to better connected agents generates many additional short connections at once, so incentives concentrate links on a dense subset. Finally, we show that the welfare maximizing network is always either empty or complete.
\end{abstract}

\section{Introduction}

In many social, economic, and information networks, an agent's payoff depends on her position in the network. Agents may value visibility, influence, or access to others, and a common way to quantify such advantages is through centrality measures. While centrality measures are commonly used in the literature to rank nodes \emph{ex post}, a basic theoretical question is what network architectures emerge when agents strategically form links \emph{ex ante} in order to increase their own centrality.

We study a network formation game in which agents sponsor links, trading off linear link maintenance costs against the benefits of being well positioned in the resulting network. We focus on a broad class of centrality measures that aggregate access along walks of all lengths, with shorter walks receiving weakly higher weight. Let \(A\) denote the adjacency matrix of the induced network and let \(\mathbf{1}\) be the all-ones vector. We write the centrality benefit of agent \(i\) as \(b_i(A)=(f(A)\mathbf{1})_i\), where \(f(z)=\sum_{k\ge 1} a_k z^k\) with \(a_k>0\) and \(a_1\ge a_2\ge\cdots\). Since \((A^k\mathbf{1})_i\) equals the total number of length-\(k\) walks that start at \(i\), this means that \(b_i(A)\) is a weighted sum of walk counts starting from \(i\). This specification is natural in environments with distance decay, and it nests prominent measures used in applications, including total communicability and Katz--Bonacich centrality.

\medskip
\noindent\textbf{Core-periphery architectures.}
A striking regularity across many empirical settings is that networks often display a core-periphery structure: a set of core nodes is densely interconnected, while the remaining nodes connect to the core but have relatively few links among themselves. Core-periphery networks have been documented in domains ranging from interbank lending markets to R\&D collaborations, making them a central object of study in both network science and economic models. In a classic contribution, \citet{HojmanSzeidl2008} derive a sharp core-periphery prediction in a homogeneous strategic model with distance decay and decreasing returns: the unique equilibrium network is a periphery-sponsored star. Their analysis provides a benchmark explanation for why strong centralization can arise endogenously even among identical agents.

\medskip
\noindent\textbf{Main result.}
Our main result shows that a general core-periphery architecture is a robust equilibrium implication of centrality maximization.
In any centrality network equilibrium, the resulting network is a core-periphery network (Theorem~\ref{theorem:main}).
That is, there exists a set of nodes $C$ such that every node is connected to all nodes in $C$, and no nodes outside of $C$ are connected to each other.

The connection between the model and the prediction is driven by two features of the environment. First, link benefits are \emph{cumulative}: a new link expands access through walks of all lengths with positive weights. In particular, linking to a well connected node does not only create a direct connection, but also creates many new short and long walks that start at the agent by routing through that node and its neighborhood. Second, because walk weights are weakly decreasing in length, the main gains from such a link come from relatively short walks, so access to a densely connected region is especially valuable. Together these features generate a structural complementarity: links are most valuable when they connect into parts of the network that are themselves dense. This rich get richer force resembles the intuition of random attachment models such as \citet{jackson2007meeting}, but here it operates through deterministic best responses rather than probabilistic meetings. The core periphery form is the minimal architecture consistent with these incentives: the core is dense because core nodes value linking to other core nodes most, and the periphery attaches to the core because connecting into the dense region yields large centrality gains at low marginal cost.

Importantly, the characterization allows for a nontrivial core, rather than a single center. In particular, the core can be empty (yielding the empty network) or all of $N$ (yielding the complete network), and for intermediate cases the core can contain multiple agents.

\medskip
\noindent\textbf{Efficiency.}
We complement the equilibrium analysis with a welfare characterization. Social welfare is the sum of the benefits enjoyed by the agents minus total costs of maintaining the links. Despite the richness of feasible architectures, the welfare maximizing network is always extreme: it is either the empty network or the complete network (Theorem~\ref{theorem:efficiecy}).
Thus, while equilibrium architecture is generically core-periphery, efficiency collapses to the two polar networks.

\subsection{Related Literature}

Analysis of networks spans a variety of fields, including economics, mathematics, physics, sociology, and computer science; see \citet{Jackson2008}. The network formation literature studies how the process by which links are formed, together with agents' objectives, shapes the resulting network architecture. In the seminal model of \citet{JacksonWolinsky1996}, links form only when both endpoints benefit, and pairwise stable and efficient networks need not coincide. As \citet{BlochJackson2006} show, the details of the link formation protocol matter: the set of pairwise stable networks can differ from the Nash equilibrium outcomes of Myerson's linking game, where players announce desired links and a link forms only if both endpoints agree. We use Nash equilibrium with unilateral link sponsorship, as in \citet{BalaGoyal2000}: one agent can pay to create a link, while the resulting connection is available to both endpoints.

A number of papers derive core-periphery structures from endogenous information or access considerations. \citet{HojmanSzeidl2008} study a network formation model in which benefits decay with graph distance and exhibit strong decreasing returns; their unique nonempty equilibrium is a periphery-sponsored star. \citet{galeotti2010law} study information acquisition and communication, where agents choose both how much information to acquire personally and whom to link to in order to access others' information. In strict equilibrium, information acquisition is concentrated among a small set of hubs, while the remaining agents link to them. \citet{HerskovicRamos2020} study directed information acquisition in a beauty-contest environment. Agents link to observe other agents' signals, and signals that are observed by more agents become more useful for predicting the average action. Under their baseline cost assumption, strict equilibrium information structures are hierarchical directed networks; with the additional assumption of weakly increasing marginal costs, strict equilibrium information structures are core-periphery directed networks.

Our contribution is complementary to these results, but starts from a different question. We ask what network architectures are forced by centrality maximization itself. This distinction is important because the result applies to every equilibrium network: under linear link costs, every Nash equilibrium of our game induces a core-periphery network. The reason is that walk-based centrality makes the value of a link cumulative. A link to a well-connected region does not only create one additional connection; it creates access to many weighted walks generated by that region. This force makes dense regions especially attractive as linking targets and rules out equilibrium architectures that are not organized around a common core. Thus, the paper identifies a centrality-based route to core-periphery formation and shows that, for a broad class of decreasing walk-based centrality objectives, core-periphery is not merely one possible outcome but a necessary equilibrium implication.

Our paper is also related to the literature connecting network position to equilibrium behavior in games played on fixed networks. Most closely, \citet{BallesterCalvoArmengolZenou2006} study a linear-quadratic game with local complementarities and show that, for a fixed network, equilibrium actions are proportional to Bonacich centralities. As discussed in Subsection~\ref{subsec:strategic-foundation}, this Bonacich--Nash linkage provides a strategic rationale for the Katz--Bonacich objective in our model. The difference is that in their analysis the network is fixed and centrality determines equilibrium behavior, whereas in our model agents form the network itself in anticipation of the continuation value generated by their position.

Our work is also related to the literature on nested split graphs and centrality-based network formation. \citet{KonigTessoneZenou2014} study a dynamic formation model based on Bonacich centrality and stability concepts, finding that nested split graphs emerge. Similarly, \citet{BelhajBervoetsDeroian2016} identify nested graphs and complete networks as efficient in games with local complementarities. Our contribution is to provide a sharp structural characterization for the entire class of walk-based measures using standard Nash equilibrium. We show that under monotonicity and decay, equilibrium networks collapse specifically to a two-tier core-periphery structure, and that the tension between equilibrium stratification and efficient integration (Theorem~\ref{theorem:efficiecy}) is a general property of this class.

Finally, a recent strand studies formation when agents maximize spectral centralities. The closest comparison is the formation game analyzed by \citet{CatalanoCastaldoComoFagnani2025}. While they frame their analysis within the Bonacich family, they specifically study PageRank, where the adjacency matrix is row-normalized and every outgoing link is chosen by the node that creates it. This normalization creates a dilution effect: connecting to a high-degree node provides less influence than connecting to a low-degree node. Consequently, they find that equilibria tend to be fragmented, featuring local cycles and disconnected components. In contrast, the class of centralities we consider, including standard Katz and communicability, is cumulative: linking to a hub provides access to walks without diluting the value of existing connections. This distinction is fundamental: while relative measures drive fragmentation, our results show that cumulative walk-based measures drive cohesion and core-periphery stratification.

\subsection{Why maximize centrality? A strategic rationale}
\label{subsec:strategic-foundation}

A simple transmission interpretation gives a foundation for the walk-based payoff. Suppose that, after the network is formed, agent \(i\) generates an opportunity, such as a job lead, a client lead, or valuable information. The opportunity is first transmitted to \(i\)'s neighbors and may then continue to circulate through the network. After \(t\) previous transmissions, each recipient forwards it to each neighbor with probability \(\rho_t\in(0,1]\). Each successful transmission, rather than each distinct recipient, generates payoff-relevant value for the originator. Thus every walk of length \(k\) starting from \(i\) represents one possible transmission event chain, possibly revisiting agents, including \(i\) herself. Its weight is \(\rho_1\cdots\rho_{k-1}\), with direct transmissions receiving weight one. Hence the expected transmission value generated by \(i\)'s position is a weighted sum of walks starting from \(i\). Conversely, any positive weakly decreasing sequence \((a_k)\) can be interpreted this way, up to normalization, by setting \(\rho_k=a_{k+1}/a_k\).

A related strategic rationale comes from games played on fixed networks. To illustrate this point, consider the Katz--Bonacich specification. Suppose that the interaction has two stages. In the first stage, agents choose which links to sponsor and pay the associated link costs. These choices determine the realized network. In the second stage, agents play a linear-quadratic network game with local complementarities on this realized network, as in \citet{BallesterCalvoArmengolZenou2006}.

The key observation is their Bonacich--Nash linkage. For every fixed network, the unique interior Nash equilibrium action of each player is proportional to her Bonacich centrality in that network. Moreover, equilibrium utility is tied to this equilibrium action, and therefore to Bonacich centrality. In the benchmark case without the aggregate substitutability term, equilibrium utility is proportional to the square of Bonacich centrality.

Thus, the point is not that the Ballester--Calvó-Armengol--Zenou model delivers exactly the linear payoff \(b_i(A)-c|s_i|\). Rather, it shows that Katz--Bonacich centrality is a payoff-relevant object in a standard strategic interaction played on a fixed network. Our payoff specification can therefore be viewed as a tractable reduced form that captures the incentive to occupy a favorable Bonacich position in the induced network, while the cost term captures the resources required to sponsor links.

Walk-based centralities also arise naturally in models of production networks. In fixed production-network models, input-output linkages determine how prices, shocks, and profits propagate through direct and indirect supply chains. For example, \citet{AcemogluCarvalhoOzdaglarTahbazSalehi2012} show that the network of input-output linkages shapes aggregate fluctuations through the propagation of firm-level shocks. Relatedly, \citet{HV} show that, in a Cobb--Douglas production economy, equilibrium profits can be represented using a Bonacich-type centrality induced by the input-output matrix. These papers take the production network as fixed. Closer to our endogenous-network perspective, \citet{MND} study firms that strategically choose input shares and show that profit maximization is equivalent to maximizing eigenvector centrality in the induced production network.

Taken together, these examples motivate centrality maximization as a reduced-form way to capture environments in which network position generates payoff-relevant access, influence, or propagation. Our analysis asks which network architectures emerge when agents form links in anticipation of the strategic value generated by their position in the network.
\section{Centrality Network Formation}\label{sec:model}

We consider a strategic model of network formation with a finite set of agents $N=\{1,\dots,n\}$. 
Each agent $i$ chooses a (possibly empty) set of agents to whom she sponsors a link,
\[
s_i \subseteq N\setminus\{i\}.
\]
We denote the strategy profile by $s=(s_1,\dots,s_n)$. The induced network is the undirected graph $G(s)=(N,E(s))$ whose edge set is
\[
E(s)\defeq \bigl\{\{i,j\}\subseteq N : i\neq j,\ \ j\in s_i \text{ or } i\in s_j\bigr\}.
\]
Thus a link $\{i,j\}$ is formed if at least one of the two agents sponsors it. Let $A=A(s)$ denote the adjacency matrix of $G(s)$, namely $A_{ij}=1$ if $\{i,j\}\in E(s)$ and $A_{ij}=0$ otherwise (with $A_{ii}=0$). 
Since $G(s)$ is undirected, $A(s)$ is symmetric. Note in particular that the benefits from an edge are shared, in the sense that connectivity is mutual, while the cost of an edge is borne by the agent(s) who sponsor it through their choice of $s_i$.

\subsection{Payoffs: Centrality and Costs}

Agents derive utility from their position in the network. We focus on a general class of centrality measures where an agent's status is determined by the total weight of all walks emanating from them.

For an adjacency matrix $A$, the term $(A^k)_{ij}$ represents the number of walks of length $k$ between agents $i$ and $j$.
The total number walks of length $k$ that start at agent $i$ is thus given by
\[
(A^k\mathbf{1})_i=\sum_{j\in N}(A^k)_{ij},
\]
where $\mathbf{1}\in\real^n$ is the vector of all ones. We assign a weight $a_k$ to walks of length $k$.
We assume that the value of a connection decays with distance, and therefore impose that the weights are positive and weakly decreasing:
\[
a_1 \ge a_2 \ge \cdots > 0.
\]

The centrality of agent $i$, denoted $b_i(A)$, is the weighted sum of all walks that start from agent $i$,
\[
b_i(A)= \sum_{k=1}^\infty a_k (A^k \mathbf{1})_i
= \left( \sum_{k=1}^\infty a_k A^k \mathbf{1} \right)_i.
\]
Agent $i$'s gross benefit from a network $G$ with adjacency matrix $A$ is given by her centrality $b_i(A)$.

We assume throughout that, for every network on $N$, the series $\sum_{k\ge 1} a_k A^k$ converges.\footnote{Equivalently, $\sum_{k\ge 1} a_k z^k$ has radius of convergence larger than $n-1$ and hence $\sum_{k\ge 1} a_k A^k$ converges for every adjacency matrix on $n$ nodes.} This formulation captures the idea that being connected to central agents is valuable because it provides access to their stock of walks (indirect connections).

Agents trade off this centrality benefit against the cost of maintaining relationships. The payoff to player $i$ under strategy profile $s$ is
\begin{align}\label{Defpayoff}
u_i(s)= b_i(A(s)) - c \, |s_i|,
\end{align}
where $c > 0$ is the constant marginal cost of sponsoring a link.\footnote{In Section~\ref{sec:convexcost} we also discuss increasing marginal costs.}

\subsection{Examples of Centrality Measures}

Our framework encompasses several prominent centrality measures used in social network analysis and economics. Each of these measures corresponds to a particular choice of the generating function $f(z)=\sum_{k\ge 1} a_k z^k$.

\paragraph{Katz--Bonacich Centrality.}
Perhaps the most widely used measure in economic models, Katz centrality assumes geometric decay of influence. We fix $\delta\in(0,1/(n-1))$,\footnote{This bound guarantees that the power series $\sum_{k\ge 0}\delta^k A^k$ converges, and hence $(I-\delta A)^{-1}$ is well defined, for every feasible adjacency matrix $A$.} so that the utility is well defined for every strategy profile. The weights are $a_k=\delta^k$, and the centrality vector corresponds to the resolvent of the adjacency matrix:
\[
b(A)=\delta A(I-\delta A)^{-1}\mathbf{1}.
\]

\paragraph{Total Communicability (Matrix Exponential).}
Based on diffusion interpretations, this measure assigns weights that decay factorially, $a_k=\beta^k/k!$, so the series converges for any adjacency matrix. The resulting centrality is given by the row sums of the matrix exponential:
\[
b(A)=e^{\beta A}\mathbf{1}.
\]
This specification is standard in the network literature, and is closely related to other matrix-function based measures; see \citet{estrada2005subgraph,EstradaHigham2010}.

\paragraph{Logarithmic Centrality.}
Used to capture logarithmic returns to distance, this measure sets weights $a_k=\alpha^k/k$. It corresponds to the matrix function:
\[
b(A)=-\log(I-\alpha A)\mathbf{1}.
\]

\paragraph{Mittag-Leffler Centrality.}
This family of measures generalizes the matrix exponential, allowing finer control over how quickly the value of a walk decays. It is defined using the Mittag-Leffler function $E_\alpha(z)=\sum_{k=0}^\infty z^k/\Gamma(\alpha k+1)$ with parameters $\alpha,\beta>0$:
\[
b_i(A)=\bigl(E_\alpha(\beta A)\1\bigr)_i
      =\sum_{k=0}^\infty \frac{\beta^k}{\Gamma(\alpha k+1)}(A^k\1)_i,
\qquad a_k=\frac{\beta^k}{\Gamma(\alpha k+1)}.
\]
See \citet{haubold2011mittag}. To ensure that the weights are decreasing ($a_1\ge a_2\ge \cdots$), the parameter $\beta$ must not be too large; specifically, we require $\beta\le \frac{\Gamma(2\alpha+1)}{\Gamma(\alpha+1)}$.

\section{Equilibrium Architecture}\label{sec:mainth}

Our main theoretical goal is to characterize the structural features of networks that emerge in equilibrium. In games of network formation, equilibria can often be complex or fragmented. We show that, under our walk-based measures with decreasing weights, equilibrium networks must take a highly ordered form.

We formally define this class of graphs as follows.

\begin{definition}[Core-periphery network]
A graph $G=(N,E)$ is a \emph{core-periphery network} if there exists a subset $C\subseteq N$ (the ``core'') such that for any distinct $i,j\in N$,
\[
\{i,j\}\in E \iff i\in C \text{ or } j\in C.
\]
Nodes in $N\setminus C$ are referred to as the ``periphery.''
\end{definition}

Equivalently, $C$ forms a clique, $N\setminus C$ forms an independent set, and every node in $N\setminus C$ is linked to every node in $C$.

This definition covers a broad spectrum of architectures depending on the size of the core. At the extremes, it captures the empty network (where $C=\varnothing$) and the complete network (where $C=N$). It also nests the star architecture (where $|C|=1$) often found in previous work. Our framework allows for the intermediate case $1<|C|<N$, representing a multi-agent core that is internally dense and to which the remaining agents attach.

Our main result states that this architecture is the unavoidable outcome of strategic centrality maximization.

\begin{theorem}\label{theorem:main}
In any centrality network equilibrium, the resulting network is a core-periphery network.
\end{theorem}

The prevalence of core-periphery structures has been extensively cataloged in network science, dating back to the formalizations by \citet{borgatti2000models}. Theorem~\ref{theorem:main} underscores that this architecture is the robust attractor for strategic agents maximizing broad influence.

By allowing for a non-trivial core ($|C| > 1$), our result bridges the gap between two common extremes found in the literature: the star network (a core of size 1) often predicted by connections models like \citet{BalaGoyal2000}, and the complete clique predicted when benefits are purely local.

This characterization has profound implications for systemic stability. As shown by \citet{acemoglu2015systemic}, architectures that are highly interconnected at the core exhibit a ``robust-yet-fragile'' property: they dampen small shocks effectively but facilitate catastrophic contagion during large shocks. Our result implies that the incentives to maximize centrality naturally drive the system toward these risk-concentrating architectures. Thus, the core-periphery structure is not just an efficient configuration for information flow, but an inevitable equilibrium outcome that may endogenously generate systemic fragility.

\begin{figure}[H]
\centering
\includegraphics[width=0.65\textwidth]{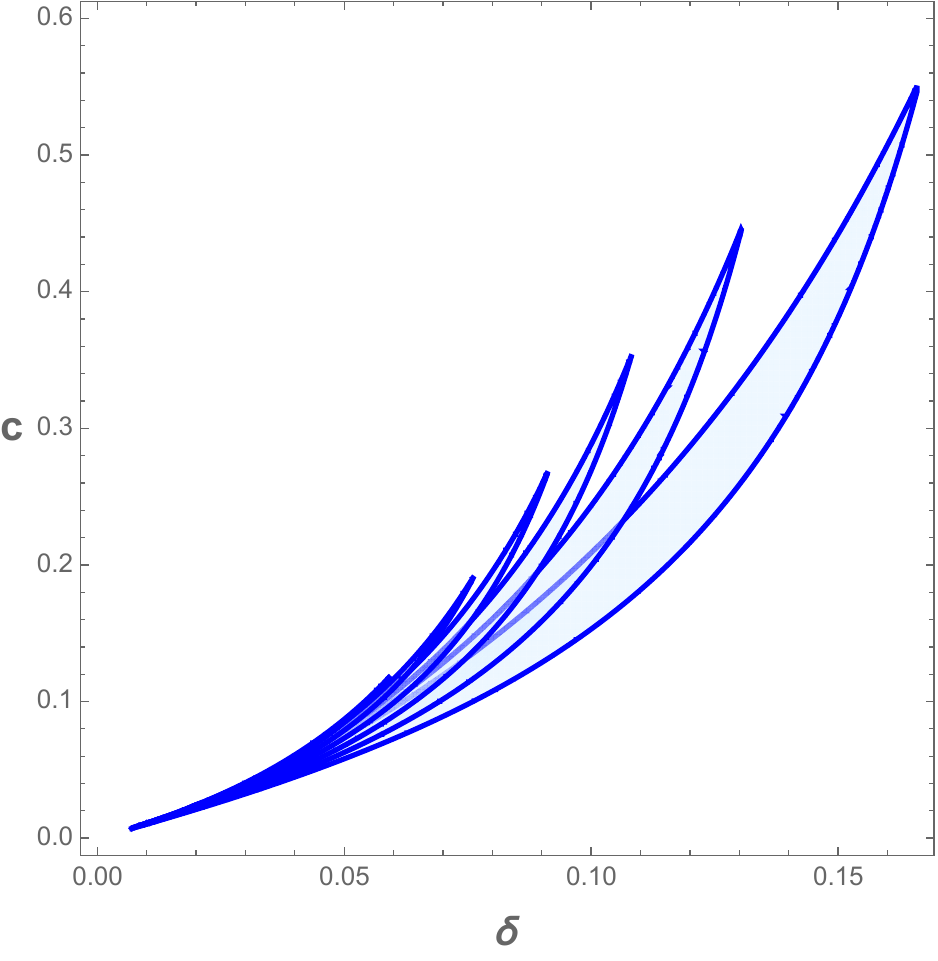}
\caption{The Core Size in Katz-Bonacich}
\label{fig:moons}
\end{figure}

As an illustration of Theorem \ref{theorem:main}, consider an example where $n=10$ agents maximize their Katz-Bonacich centrality (benefit from walks decays geometrically with length), $b(A)=\delta A(I-\delta A)^{-1}\mathbf{1}$.

Figure \ref{fig:moons} plots the parameter region $(\delta, c)$ where a non-trivial core-periphery network is an equilibrium. Each moon-shaped region corresponds to a given size of the core, ranging from $|C|=1$ (the rightmost) to $|C|=8$ (the leftmost).
The regions in the figure are computed using the characterization in Section~\ref{sec:core-feasibility}: for each core size \(k\), equilibrium requires that peripheral agents neither delete their links to the core nor add all missing links to the periphery.
\section{Proof of Theorem \ref{theorem:main}}

The proof of Theorem~\ref{theorem:main} relies on a structural complementarity argument: because link benefits are walk-based, adding an edge becomes more valuable when it connects to agents who already have many paths to the rest of the network, and this effect is stronger at shorter distances (since weights decay with length). As a result, incentives push links to concentrate on a dense subset of agents.

We rule out any architecture that is not core-periphery by exhibiting a profitable deviation. If the network is not core-periphery, we can find an agent who gains by dropping some of her current links and redirecting them toward a denser region of the network, thereby increasing the value of her walk-based benefits.

To formalize the deviation argument, we first establish a condition for when one set of edges yields higher centrality than another.

Let $G=(N,E)$ be a network. For any subset of edges $\mathcal{E} \subseteq E$ and agent $i$, let $\Psi^{\mathcal{E}}_{G,i}$ denote the set of all finite walks in $G$ that start at $i$ and traverse at least one edge in $\mathcal{E}$. 
For any walk $\gamma=(v_0,\dots,v_T)$, denote by $\ell(\gamma)$
the length of $\gamma$.
We define the incremental centrality agent $i$ derives from $\mathcal{E}$ as the weighted sum of these walks:
\[
b^{\mathcal{E}}_{G,i} = \sum_{\gamma \in \Psi^{\mathcal{E}}_{G,i}} a_{\ell(\gamma)}.
\]

If an agent $i$ plays a best response, she cannot be sponsoring a link that is already sponsored by another agent; i.e. $j \in s_i \implies i \notin s_j$. In addition, two properties must hold regarding the marginal cost $c$:
\begin{enumerate}
    \item[(i)] \textbf{Retention:} For any subset of agents $T \subseteq s_i$ that $i$ sponsors links to, the benefit must exceed the cost: $b^{E_i}_{G,i} \ge |E_i|c$, where $E_i$ are the edges connecting $i$ to agents in $T$.
    \item[(ii)] \textbf{No Deviation:} For any other set of agents $T\subseteq N\setminus(\{i\}\cup s_i)$ that $i$ does \emph{not} sponsor links to, adding them must not be profitable: $b^{E_i}_{G',i} \le |E_i|c$, where $G'$ is the graph with the added edges.
\end{enumerate}

We compare these benefits using a structural mapping argument. Loosely speaking, Lemma~\ref{lemma:deviation} shows that if agent $i$ can mimic the connectivity of a set of other agents $j_1,\dots,j_k$ in a way that ``covers'' their neighbors, then she derives strictly more benefit than they do from the relevant connections. This implies that either agent $i$ has a profitable deviation in the form of mimicking these agents' connectivity, or at least one of the other agents has a profitable deviation in the form of reducing their connectivity.

To formalize this deviation argument, we first establish a condition for when one set of edges yields higher centrality than another.

Let $G=(N,E)$ be a graph and $i \in N$ an agent. Consider any collection of distinct agents $j_1, \dots, j_k \in N \setminus \{i\}$ and associated edge sets $E_1, \dots, E_k$ satisfying the following four conditions:
\begin{enumerate}
    \item $E_l \subseteq E$ and $j_l \in e$ for every $e \in E_l$.
    \item $E_l \cap E_m = \varnothing$ for distinct $l, m$.
    \item  If $\{j_l, f\} \in E_l$, then $f \notin N_G(i) \cup \{i\}$.
    \item If $\{j_l, f\} \in E_l$ and $\{j_m, g\} \in E_m$ with $l \neq m$, then $f \neq g$.
\end{enumerate}
Informally, agent $i$ adds one link to each neighbor that any of the agents $j_1,\dots,j_k$ reaches through the designated edge sets $E_1,\dots,E_k$, thereby covering all of those connections with the same total number of new links. Formally, let $G'$ be the graph obtained by adding the set of edges $E_i=\{\{i,f\}:\exists l\text{ s.t. }\{j_l,f\}\in E_l\}$ to $G$. By Condition~4 these edges are distinct and $|E_i|=\sum_{l=1}^k|E_l|$. Let $\Psi=\bigcup_{l=1}^k\Psi^{E_l}_{G,j_l}$ denote the set of all walks that start at some $j_l$ and traverse at least one edge in $E_l$.

\begin{restatable}{lemma}{injectivelem}
\label{lemma:deviation}
If $N_{G'}(j_l)\setminus\{i\}\subseteq N_{G'}(i)\setminus\{j_l\}$ for all $l=1,\dots,k$, then there exists an injective, non-surjective mapping $F:\Psi\to\Psi^{E_i}_{G',i}$ such that $\ell(\gamma)\ge \ell(F(\gamma))$ for all $\gamma\in\Psi$.
\end{restatable}

\begin{proof}
See Appendix.
\end{proof}



\begin{proposition}\label{prop:profit}
Let $s$ be a strategy profile and $G = G(s)$ be the corresponding network. Suppose there exists an agent $i$, a target set of agents $T \subseteq N \setminus (N_G(i) \cup \{i\})$ to whom $i$ is not connected, and a reference set $K \subseteq N \setminus \{i\}$ such that the following hold:
\begin{enumerate}
    \item Every agent in $T$ has a connection sponsored by some agent in $K$:
    \begin{align*}
        \text{for every } j \in T, \, j \in s_l \text{ for some } l \in K
    \end{align*}
    \item Forming edges to the agents in $T$ makes $i$ more connected than every agent in $K$:
    \begin{align*}
        N_{G'}(l) \setminus \{i\} \subseteq N_{G'}(i) \setminus \{l\} \text{ for every } l \in K, \text{ where } G' = G(s_i \cup T, s_{-i})
    \end{align*}
\end{enumerate}
Then $s$ is not an equilibrium.
\end{proposition}
In particular, we show that if all of the $j_l$ are best responding, then agent $i$ has a profitable deviation by adding the links $E_i$.

\begin{proof}
Assume by way of contradiction that $s$ is an equilibrium. We decompose the target set $T$ based on who currently sponsors the links.
We select agents $j_1, \dots, j_k$ from $K$ and identify disjoint subsets $T_1, \dots, T_k$ of $T$ such that $T_r \subseteq s_{j_r}$.
Specifically, choose $j_1 \in K$ such that $T_1 = s_{j_1} \cap T$ is maximal. If $T_1 \neq T$, choose $j_2$ to cover a maximal portion of the remainder $T \setminus T_1$, let $T_2=(s_{j_2}\cap T)\setminus T_1$ and proceed inductively until $\bigcup T_r = T$ where the union is over disjoint sets. Let $E_r = \{\{j_r, f\} : f \in T_r\}$.
This construction creates disjoint sets of edges $E_1, \dots, E_k$ satisfying the disjointness conditions of Lemma \ref{lemma:deviation}.

Consider the deviation where $i$ adds connections to all agents in $T$. Let $E_i$ be these new edges. By the premise of the proposition, $N_{G'}(l) \setminus \{i\} \subseteq N_{G'}(i) \setminus \{l\}$. Therefore, 
Lemma \ref{lemma:deviation} applies, yielding:
\begin{align}
 &\notag \sum_{r=1}^k b^{E_r}_{G,j_r}
=
\sum_{r=1}^k\sum_{\gamma\in\Psi^{E_r}_{G,j_r}}a_{\ell(\gamma)}
\leq\\
&\label{eq:deviation1}
\sum_{r=1}^k\sum_{\gamma\in\Psi^{E_r}_{G,j_r}}a_{\ell(F(\gamma))}
=
\sum_{\gamma'\in F(\Psi)}a_{\ell(\gamma')}
< \\ 
&\label{eq:deviation2} 
\sum_{\gamma'\in\Psi^{E_i}_{G',i}}a_{\ell(\gamma')}
=
b^{E_i}_{G',i}.
\end{align}
Inequality \eqref{eq:deviation1} follows since
$\ell(\gamma)\geq \ell(F(\gamma))$ by Lemma~\ref{lemma:deviation}
and the sequence $(a_t)$ is weakly decreasing. Inequality
\eqref{eq:deviation2} follows since $F$ is injective and non-surjective,
so $F(\Psi)\subsetneq \Psi^{E_i}_{G',i}$, and all weights are positive. Therefore,
\begin{equation}\label{eq:strict_gain}
    \sum_{r=1}^k b^{E_r}_{G,j_r} < b^{E_i}_{G',i}.
\end{equation}
Since $s$ is an equilibrium, the agents $j_r$ must find it profitable to maintain their links. Thus, by the Retention property, $b^{E_r}_{G,j_r} \ge |E_r|c$.
Summing over $r$:
\[
\sum_{r=1}^k |E_r|c \le \sum_{r=1}^k b^{E_r}_{G,j_r} < b^{E_i}_{G',i}.
\]
Since $|E_i| = \sum |E_r| = |T|$, we have $|E_i|c < b^{E_i}_{G',i}$.
This implies that for agent $i$, the benefit of adding links $T$ strictly exceeds the cost. This contradicts the No Deviation condition required for equilibrium.
\end{proof}

We next turn to the proof of Theorem \ref{theorem:main}. 
\begin{definition}
Consider the network formation game defined above and let $s=(s_1,\ldots,s_n)$ be a strategy profile. We call agent $j$ isolated if $j\not\in s_i$ for any $i\neq j$. That is $j$ is isolated if no agent $i$ sponsor a link to $j$.    
\end{definition}

\begin{proof}[\textbf{Proof of Theorem \ref{theorem:main}}]
~\\
\noindent\textbf{Case 1: No isolated agents.}
Suppose every agent receives at least one sponsored link. We claim that $G$ must be complete. Assume not. Then there exist an agent $i$ and a nonempty set of non-neighbors $T = N \setminus (N_G(i) \cup \{i\})$. Let $K = N \setminus \{i\}$.
If $i$ connects to $T$, she becomes connected to everyone (so $N_{G'}(i) = N \setminus \{i\}$). Consequently, for any $l \in K$, we have $N_{G'}(l) \setminus \{i\} \subseteq N_{G'}(i) \setminus \{l\}$.
Since there are no isolated agents, the conditions of Proposition \ref{prop:profit} are satisfied, implying that $s$ is not an equilibrium. Thus, if no agent is isolated, $G$ is complete (a core-periphery network with $C=N$).

\medskip
\noindent\textbf{Case 2: Isolated agents exist.}
Assume there is at least one isolated agent $j$ (that is, $s_l$ does not contain $j$ for any $l$). Let $C = s_j$. We show that $C$ acts as the core.\\
\emph{Step A: All agents are connected to $C$.}
If $C=\varnothing$, there is nothing to prove. Suppose $C\neq\varnothing$ and there is an agent $i$ who is not fully connected to $C$. Let
\[
T=C\setminus (N_G(i)\cup\{i\}).
\]
Then $T$ is nonempty. Let $K=\{j\}$. Since $C=s_j$, every agent in $T$ has a connection sponsored by $j$. If $i$ connects to $T$, then in the resulting graph $G'$ we have
\[
N_{G'}(j)\setminus\{i\}\subseteq N_{G'}(i)\setminus\{j\}.
\]
Thus the conditions of Proposition~\ref{prop:profit} are satisfied, implying that $s$ is not an equilibrium, a contradiction. Therefore,
\[
C\setminus\{i\}\subseteq N_G(i)
\]
for every agent $i$.\\
\emph{Step B: No peripheral edges.}
Assume next that an isolated node $j$ exists with respect to $s$. Let $C=s_j$ (the set $C$ may be empty). We claim that all agents $i$ are connected to $C$, that is, $C\setminus\{i\}\subseteq N_G(i)$. Otherwise let $T=C\setminus (s_i\cup\{i\})$
and $K=j$. Note that the conditions of Proposition \ref{prop:profit} are satisfied with respect to $K$ and $T$. Thus, $s$ is not an equilibrium. Therefore, all agents $i$ are connected to $C$.

Note that for the same reason as in Step A, every isolated agent $i$ satisfies
$C\setminus\{i\}\subseteq s_i.$

We claim that there are no other edges $\{j,l\}$ in $G$ for agents $j,l\in N\setminus C$. Assume otherwise. Consider the following deviation of an isolated agent $i$. Let
$$T=\{h\in N\setminus C:\exists l\in N\setminus C \text{ such that } h\in s_l\}.$$
Let $K=(N\setminus C)\setminus\{i\}$. We claim that the conditions of
Proposition~\ref{prop:profit} are satisfied with respect to $i$, $K$,
and $T$. By the definition of $T$, every agent in $T$ has a connection
sponsored by some agent in $K$. Moreover, after the deviation to
$s'_i=s_i\cup T$, agent $i$ is connected to $C$ and to every agent in
$N\setminus C$ that receives a sponsored link from some agent in
$N\setminus C$, so the neighborhood inclusion required by
Proposition~\ref{prop:profit} holds. Thus Proposition~\ref{prop:profit}
applies, so $s$ is not an equilibrium, a contradiction. Hence there are
no edges between two agents in $N\setminus C$. Since Step A showed that
every agent is connected to every agent in $C$, we conclude that
$\{i,j\}\in E$ if and only if either $i\in C$ or $j\in C$. Therefore
$C$ is the core set and $G$ is a core-periphery network.
\end{proof}
We get the following corollary from the proof of Theorem~\ref{theorem:main}.

\begin{corollary}
\label{cor:two-isolated-peripheral-agents}
In any non-complete equilibrium, at least two peripheral agents are isolated in the sponsorship sense.
\end{corollary}

\begin{proof}
By Case 1 in the proof of Theorem~\ref{theorem:main}, any non-complete equilibrium has at least one isolated agent. Let \(C\) be the core identified in the proof of Theorem~\ref{theorem:main}, and let \(P=N\setminus C\). If \(C=\varnothing\), then the graph is empty and every agent is isolated, so the claim is immediate. Suppose \(C\neq\varnothing\).

Assume, toward a contradiction, that there is exactly one isolated peripheral agent, denoted by \(i\). Then every other peripheral agent is sponsored by some core agent. Let \(T=P\setminus\{i\}\), and let \(K\subseteq C\) be the set of core agents who sponsor agents in \(T\). If \(i\) adds links to all agents in \(T\), then \(i\) becomes connected to every agent in \(N\). Hence, in the resulting graph \(G'\), for every \(l\in K\),
\[
N_{G'}(l)\setminus\{i\}\subseteq N_{G'}(i)\setminus\{l\}.
\]
The conditions of Proposition~\ref{prop:profit} are therefore satisfied, so the original profile cannot be an equilibrium. This contradiction proves the result.
\end{proof}

\begin{remark}
We conjecture that a stronger statement holds: in every non-complete equilibrium, all core-periphery links are sponsored by the peripheral endpoint. Equivalently, every peripheral agent is isolated in the sponsorship sense.
\end{remark}

\section{Efficiency}\label{sec:efficiency}
In addition to determining which network structures can be supported in equilibrium, a central object in the study of network formation models is determining which networks are efficient, in the sense that they maximize social welfare. In our model, the answer turns out to be particularly clean. 
\begin{theorem}\label{theorem:efficiecy}
In any centrality model, the network that maximizes social welfare is either the empty network or the complete network.
\end{theorem}
One helpful observation is that given a network $G$ with adjacency matrix $A$, the social welfare that it induces is given by
\[
W_G=\sum_{\gamma}a_{\ell(\gamma)}-|E|c=\sum_{k=1}^{\infty}a_k\mathbf{1}^t A^k\mathbf{1}-|E|c,
\]
where the first summation is over all walks $\gamma$ in $G$.

We prove Theorem \ref{theorem:efficiecy} by showing that if (i) $G$ is socially beneficial ($W_G > 0$), (ii) the removal of any set of links in $G$ reduces social welfare, and (iii) there are two agents $i,j$ such that $j$ is connected to at least one agent that $i$ is not connected to, then the social welfare can be increased by connecting $i$ to every agent that is connected to $j$. As we show, this implies that if the social welfare is strictly positive for any nonempty network $G$ (i.e. if the empty network is inefficient), then it is maximized when $G$ is the complete network.

To state the result we will use for the first part of the proof, some additional notation is required. For a subset $\mathcal{E}\subset E$ of edges in $G$, we denote by $\Psi^{\mathcal{E}}_G$ the set of all finite walks in $G$ that use at least one edge in $\mathcal{E}$. We say the welfare associated with $\mathcal{E}$ in $G$ is
\[b^{\mathcal{E}}_G=\sum_{\gamma\in\Psi^{\mathcal{E}}_G}a_{\ell(\gamma)}.\]

We will say that $G = (N, E)$ is deletion-stable if the removal of any links reduces social welfare; that is, if for every subset $\mathcal{E} \subseteq E$ of links in $G$, $W_G' \leq W_G$, where $G' = (N, E \setminus \mathcal{E})$ is the network formed by removing the links in $\mathcal{E}$. Clearly, any network which maximizes the social welfare must be deletion-stable.

Denote by $G^{i,j}$ be the network formed by adding a link between $i$ and every agent $k$ which is connected to $j$ but not $i$.

\begin{restatable}{proposition}{auxprop}
\label{props:aux2}
Suppose $G$ is deletion-stable. Then for any agents $i$ and $j$, if $G^{i, j} \neq G$ then $W_G < W_{G^{i, j}}$.
\end{restatable}
That is, if $G$ has the property that its welfare cannot be improved by removing any subset of the links, then adding links between $i$ and all neighbors of $j$ (assuming $i$ is not already linked to all of $j$'s neighbors) strictly increases the welfare. We relegate the proof of the proposition to an appendix.




We next derive Theorem \ref{theorem:efficiecy} as a corollary of Proposition \ref{props:aux2}.
\begin{proof}[\textbf{Proof of Theorem \ref{theorem:efficiecy}}]
 Let $G$ be a network that maximizes social welfare. Then $G$ is deletion-stable, since $W_G \geq W_{G'}$ for all networks $G'$. By Proposition~\ref{props:aux2}, it follows that for every pair of agents $i, j$, $G^{i, j} = G^{j, i} = G$. The only networks that have this property are the empty network and the complete network.
 
 To see this, note that if an agent $i$ has at least one connection, say to $j$, then for every other agent $k$, $G^{k, j} = G$, so $k$ is connected to $i$. Thus, any agent which has at least one connection in $G$ is connected to all other agents. If there is at least one such agent, then all agents have at least one connection, so all agents must be connected to all other agents. Hence, if $G$ is not the empty network, then $G$ must be the complete network.

\end{proof}
\nocite{*}
\section{Supermodularity}\label{sec:supermodularity}
Supermodular games, also called games with strategic complementarities, are useful because they cover many applied models and because their monotonicity structure delivers clean existence results and sharp comparative statics. In particular, supermodularity implies that best responses move in the same direction, and this often yields extremal equilibria and well behaved adjustment dynamics. 

In this section, we show that the game is supermodular. We begin by fixing notation. Given a set $N$ of agents, a graph $G$ with $N$ as its vertices, and a finite sequence $p = (v_0, \dots, v_m)$ of agents, say that $p$ is a path in $G$, and write $p \in G$, if $v_i$ is connected to $v_{i+1}$ in $G$ for each $i = 0, \dots, m-1$.

A strategy $s_i$ for agent $i$ is a subset of the other agents $s_i \subseteq N \setminus \{i\}$. The network $G(s)$ is formed when agents use the strategy profile $s$: $i$ is connected to $j$ if either $j \in s_i$ or $i \in s_j$ (or both).

Denote by $P^i$ the set of finite sequences of walks that begin with agent $i$. Observe that the utility of agent $i$, defined in \eqref{Defpayoff} is given by
\begin{align*}
    u_i(s) = \sum_{\gamma \in P^i} \indic_{\gamma \in G(s)} a_{\ell(\gamma)} - c \cdot |s_i|.
\end{align*}

\begin{proposition}
\label{prop:supermodular-utility}
$u_i(s)$ is supermodular in $s_i$.
\end{proposition}

To prove this, we first isolate the basic set theoretic step. Fix a walk $\gamma$. The next lemma shows that the indicator of whether $\gamma$ is contained in a graph is supermodular as a function of the edge set: containment in $E\cup E'$ plus containment in $E\cap E'$ dominates containment in $E$ and in $E'$ separately.

\begin{lemma}
\label{lem:path-containment-supermodular}
For any sets $E, E'$ and any walk $\gamma$,
\begin{align*}
    \indic_{\gamma \in G(E \cup E')} + \indic_{\gamma \in G(E \cap E')} \geq \indic_{\gamma \in G(E)} + \indic_{\gamma \in G(E')}.
\end{align*}
\end{lemma}

\begin{proof}
Observe first that the only possible values $\indic_{\gamma \in G(E)} + \indic_{\gamma \in G(E')}$ can take are $0$, $1$, and $2$.

If $\indic_{\gamma \in G(E)} + \indic_{\gamma \in G(E')} = 0$, then this inequality holds trivially, since $\indic_{\gamma \in G(E \cup E')} + \indic_{\gamma \in G(E \cap E')} \geq 0$.

If $\indic_{\gamma \in G(E)} + \indic_{\gamma \in G(E')} = 1$, then either $\gamma \in G(E)$ or $\gamma \in G(E')$, and in either case $\gamma \in G(E \cup E')$, so 
\begin{align*}
    \indic_{\gamma \in G(E \cup E')} + \indic_{\gamma \in G(E \cap E')} \geq \indic_{\gamma \in G(E \cup E')} = 1.
\end{align*}

If $\indic_{\gamma \in G(E)} + \indic_{\gamma \in G(E')} = 2$, then $\gamma$ does not contain any edges from $E \setminus E'$, since then $\gamma \notin G(E')$, or $E' \setminus E$, since then $\gamma \notin G(E)$. It follows that $\gamma$ contains only edges from $E \cap E'$, so $\gamma \in G(E \cap E')$ and $\gamma \in G(E \cup E')$, and hence $\indic_{\gamma \in G(E \cup E')} + \indic_{\gamma \in G(E \cap E')} = 2$.
\end{proof}

The proof of supermodularity is now a straightforward application of the Lemma.

\begin{proof}[Proof of Proposition~\ref{prop:supermodular-utility}]
Fix a strategy $s_{-i}$ of the other agents, and for any $Z \subseteq N \setminus \{i\}$, denote by $E_Z$ the set of edges in $G(Z, s_{-i})$. Observe that for any $X, Y \subseteq N \setminus \{i\}$, $E_{X \cup Y} = E_X \cup E_Y$ and $E_{X \cap Y} = E_X \cap E_Y$, so by Lemma~\ref{lem:path-containment-supermodular}, for any $\gamma \in P^i$,
\begin{align*}
    \indic_{\gamma \in G(X \cup Y, s_{-i})} + \indic_{\gamma \in g(X \cap Y, s_{-i})} &= \indic_{\gamma \in G(E_{X \cup Y})} + \indic_{\gamma \in G(E_{X \cap Y})}\\ &\geq \indic_{\gamma \in G(E_X)} + \indic_{\gamma \in G(E_Y)}\\ &= \indic_{\gamma \in G(X, s_{-i})} + \indic_{\gamma \in G(Y, s_{-i})}.
\end{align*}
Thus, $\indic_{\gamma \in G(s_i, s_{-i})}$ is supermodular in $s_i$. Now, since $-c \cdot |s_i|$ is also supermodular in $s_i$, and since any positive linear combination of supermodular functions is supermodular, it follows that 
\begin{align*}
    u_i(s_i, s_{-i}) = \sum_{\gamma \in P^i} a_{\ell(\gamma)} \cdot \indic_{\gamma \in G(s_i, s_{-i})} + (-c \cdot |s_i|)
\end{align*}
is supermodular in $s_i$.
\end{proof}
The next proposition is a simple implication of supermodularity. It provides a useful benchmark: at least one of the two extreme networks, the empty network or the complete network, can be supported in equilibrium.
\begin{proposition}\label{prop:extreme-network}
In every network centrality game, either the empty network or the complete network is an equilibrium network.
\end{proposition}

\begin{proof}
Let $s$ be any profile such that for every $i, j$, either $i \in s_j$ or $j \in s_i$ but not both; that is, $s$ is a profile for which $G(s)$ is the complete network and every edge is sponsored by exactly one agent. Let $s^\varnothing$ be the profile in which every agent plays $\varnothing$; this is the unique profile for which $G(s^\varnothing)$ is the empty network. We will show that either $s$ is an equilibrium or $s^\varnothing$ is an equilibrium.

To this end, suppose that $s^\varnothing$ is not an equilibrium. We begin by showing that $s_i = N\setminus\{i\}$ must be a best response to $s_{-i}^{\varnothing}$. Let $S \subseteq N \setminus \{i\}$ be any best response for $i$ that is maximal with respect to inclusion. Suppose $S \neq N \setminus \{i\}$. Since $s^\varnothing$ is not an equilibrium, $S \neq \varnothing$. Fix some $k \in S$ and $j \notin S \cup \{i\}$. Then by supermodularity,
\begin{align*}
    u_i(S \cup \{j\}, s_{-i}^\varnothing) - u_i(S, s_{-i}^\varnothing) 
    \geq 
    u_i((S \setminus \{k\}) \cup \{j\}, s_{-i}^\varnothing) 
    - u_i(S \setminus \{k\}, s_{-i}^\varnothing).
\end{align*}
By symmetry and the fact that $|(S \setminus \{k\}) \cup \{j\}| = |S|$, $u_i((S \setminus \{k\}) \cup \{j\}, s_{-i}^\varnothing) = u_i(S, s_{-i}^\varnothing)$, and since $S$ is a best response, $(S \setminus \{k\}) \cup \{j\}$ is also a best response, so the right hand side is nonnegative. But then it follows that $u_i(S \cup \{j\}, s_{-i}^\varnothing) \geq u_i(S, s_{-i}^\varnothing)$, so $S \cup \{j\}$ is a best response which contains $S$ as a strict subset, contradicting the maximality of $S$. It follows that $s_i = N \setminus \{i\}$ is a best response to $s_{-i}^\varnothing$.

We will now use this to show that $s$ must be an equilibrium. Let $s_j' = \{i\}$ for $j \neq i$. It follows from the proof of Proposition~\ref{prop:supermodular-utility} that $u_i$ is supermodular in $s_j$ for any $j$, so for any $S \subseteq s_i$,
\begin{align*}
    u_i(s_i, s_{-i}) - u(S, s_{-i}) &\geq u_i(s_i, s_{-i}') - u_i(S, s_{-i}')\\ &= u_i(N \setminus i, s_{-i}^\varnothing) - u_i(N \setminus (s_i \cup i) \cup S, s_{-i}^\varnothing) \geq 0.
\end{align*}
Moreover, since $u_i(T, s_{-i}) \leq u_i(T \cap s_i, s_{-i})$ for any $T \subseteq N \setminus i$, it follows that $u_i(s_i, s_{-i}) \geq u_i(T, s_{-i})$, so $s_i$ is a best response to $s_{-i}$. Since $i$ was arbitrary, this holds for all agents, so $s$ is an equilibrium.

The following counterexample shows that, in contrast to Theorem \ref{theorem:main}, supermodularity of $u_i(s)$ in $s_i$ alone does not imply that all equilibria are core-periphery. Sponsoring a node is more profitable for nodes with incoming links than for nodes without incoming links. As a result, there may exist equilibria in which a node is sponsored by some nodes but not by all nodes, in contrast to the core-periphery structure.

\begin{example}
Suppose $n=4$ and let $u_i(s) = z_i^\alpha - c |s_i|$, where $z_i = \bigl|
\bigl\{ j \in N_G(i) : \deg(j) > 1 \bigr\} \bigr|$ is the number of neighbors of node $i$ who have more than one neighbor, and $\alpha \in (1,2)$. Since $\alpha \geq 1$, the utility function is supermodular in $s_i$ (note: $z_i$ and $c |s_i|$ are modular).

Consider a candidate equilibrium consisting of a `triangle' plus an isolated node: $s_1 = \left\{ 2 \right\}$, $s_2 = \left\{ 3 \right\}$, $s_3 = \left\{ 1 \right\}$, and $s_4 = \emptyset$. The corresponding payoffs are: $b_1^*=b_2^*=b_3^*=2^\alpha-c$, and $b_4^* = 0$. Let us show that this is a strict Nash equilibrium for some $c$.

Observe that node 1 (the same reasoning applies to nodes 2 and 3) does not gain from sponsoring a link to node 4, because node 4 would have only one neighbor. The other possible deviation is to cut the link to node 2. The resulting payoff is $b_1^{dev}=1$, and thus the no-deviation condition is $c<2^{\alpha}-1$.

The isolated node can deviate by sponsoring links to all other nodes (if this deviation is not beneficial, then sponsoring links to only one or two nodes is also not beneficial). The resulting payoff is $b_4^{dev}=3^\alpha-3 c$, and thus the no-deviation condition is $c>3^{\alpha-1}$.

All no-deviation conditions are strictly satisfied if and only if $c \in\left( 3^{\alpha-1} , 2^{\alpha}-1 \right)$. This interval is nonempty if and only if $\alpha \in (1,2)$.
\end{example}
\end{proof}
\section{Feasibility of Core Sizes}
\label{sec:core-feasibility}

Theorem~\ref{theorem:main} shows that every equilibrium network is a core-periphery network. A natural next question is which core sizes can be supported in equilibrium. In this section, we give a characterization in terms of two simple deviations: a peripheral agent deleting all links to the core, and a peripheral agent adding links to all other peripheral agents.

Fix a set $C\subseteq N$ with $|C|=k$, and let $P=N\setminus C$. Let $G_k$ denote the core-periphery graph with core $C$: agents in $C$ are connected to all other agents, agents in $P$ are connected to all agents in $C$, and there are no edges between agents in $P$. We allow $k=0$, in which case $G_k$ is the empty graph. We use the convention that the complete graph is treated as the case $k=n$; thus the non-complete cases are $0\leq k\leq n-2$.

For \(0\leq k\leq n-2\), consider the following canonical sponsorship profile that generates \(G_k\). Each peripheral agent sponsors her links to all core agents, and no core agent sponsors a link to a peripheral agent. Each edge between two core agents is sponsored by exactly one of these two core agents. Thus, every peripheral agent is isolated in the sponsorship sense: no other agent sponsors a link to her. For a peripheral agent \(p\in P\), let \(G_k^{-p}\) be the graph obtained from \(G_k\) by deleting all edges between \(p\) and \(C\), and let \(G_k^{+p}\) be the graph obtained from \(G_k\) by adding all edges between \(p\) and the agents in \(P\setminus\{p\}\).
\begin{proposition}
\label{prop:k-core-feasibility}
Fix $0\leq k\leq n-2$. A non-complete core-periphery equilibrium with core size $k$ exists if and only if, for a peripheral agent $p\in P$,
\[
b_p(G_k)-b_p(G_k^{-p})\geq kc
\]
and
\[
b_p(G_k^{+p})-b_p(G_k)\leq (n-k-1)c.
\]
\end{proposition}
The characterization above says that, for a fixed non-complete core-periphery architecture, equilibrium feasibility is governed by the incentives of a representative peripheral agent. In the canonical sponsorship profile, peripheral agents sponsor all links to the core and no links to other peripheral agents. Thus, a peripheral agent has two natural extreme deviations. She can drop all of her sponsored links to the core, thereby becoming disconnected from the core; or she can add all missing links to the other peripheral agents, thereby making herself connected to everyone. The first inequality rules out the deletion deviation, while the second rules out the addition deviation. Supermodularity then implies that if these two extreme deviations are not profitable, then no partial deletion or partial addition is profitable either. The remaining deviations of core agents are ruled out by comparing the value of core-core links to the corresponding links from a peripheral agent to the core.
\begin{proof}
We first prove necessity. Suppose that an equilibrium with core size $k$ exists, where $0\leq k\leq n-2$.  By Corollary~\ref {cor:two-isolated-peripheral-agents}, there exists an isolated agent in the sponsorship sense. By the proof of Theorem~\ref{theorem:main}, if $j$ is such an isolated agent, then $s_j$ is the core. Hence, in a core-periphery equilibrium with core $C$, an isolated peripheral agent sponsors all links to the core.

Since this isolated peripheral agent is best responding, she cannot profit by deleting all of her links to the core. This gives the first inequality. She also cannot profit by adding links to all other peripheral agents. This gives the second inequality.

We now prove sufficiency. Consider the canonical sponsorship profile described above. We show that no agent has a profitable deviation.

First consider a peripheral agent $p$. Her possible deviations consist of deleting some of her links to the core, adding some links to other peripheral agents, or doing both. By Proposition~\ref{prop:supermodular-utility}, her payoff is supermodular in the set of links she sponsors. Since all core agents are symmetric from the point of view of $p$, if deleting all links to the core is not profitable, then deleting only some of them is not profitable. Similarly, since all peripheral agents are symmetric from the point of view of $p$, if adding all missing peripheral links is not profitable, then adding only some of them is not profitable.

It remains to rule out deviations that both delete core links and add peripheral links. This follows from supermodularity as well. The gain from adding peripheral links is weakly larger when $p$ keeps more links to the core. Hence, if adding any set of peripheral links is not profitable when $p$ keeps all of her core links, it is also not profitable after she deletes some of her core links. Combining this with the previous paragraph shows that no peripheral agent has a profitable deviation.

Next consider a core agent. A core agent is already connected to every other agent, and therefore cannot profitably add links. Thus the only possible deviation is to delete some of the core-core links she sponsors. Suppose a core agent $i$ sponsors a set $E_i$ of core-core edges and that deleting these edges is profitable. Then $b^{E_i}_{G_k,i}<|E_i|c$.

Let $p$ be a peripheral agent, and let $E_p$ be the set of edges from $p$ to the same core agents that are reached from $i$ through the edges in $E_i$. By the same walk comparison used in the proof of Proposition~\ref{prop:profit}, in particular the comparison in equations \eqref{eq:deviation1}--\eqref{eq:deviation2}, the partial benefit of these links for the peripheral agent is no larger than the partial benefit of the corresponding sponsored links for the core agent:
\[
b^{E_p}_{G_k,p}\leq b^{E_i}_{G_k,i}.
\]
Hence $b^{E_p}_{G_k,p}<|E_p|c$, since $|E_p|=|E_i|$. Thus the peripheral agent $p$ would profitably delete the links $E_p$, contradicting the fact that no peripheral agent wants to delete any subset of links to the core. Therefore no core agent wants to delete any sponsored core-core links.

Therefore no agent has a profitable deviation, and the canonical sponsorship profile is a Nash equilibrium.
\end{proof}
The complete graph is the remaining boundary case. Let \(G_n\) be the complete graph. A complete sponsorship profile is a profile that generates \(G_n\) and in which every edge is sponsored by exactly one of its endpoints. For \(d\in\{0,\ldots,n-1\}\), let \(G_n^{-d}\) denote the graph obtained from \(G_n\) by deleting \(d\) links incident to a fixed agent. By symmetry, the value \(b_i(G_n)-b_i(G_n^{-d})\) depends only on \(d\), and not on the identity of the agent or on the particular \(d\) links deleted.

\begin{proposition}
\label{prop:complete-feasibility}
Let \(h=\lceil (n-1)/2\rceil\). The complete graph is an equilibrium network if and only if
\(b_i(G_n)-b_i(G_n^{-h})\ge hc\).
\end{proposition}

\begin{proof}
Suppose first that the complete graph is an equilibrium network. In any equilibrium profile generating \(G_n\), no edge is sponsored by both endpoints, since duplicate sponsorship changes no benefits and only adds cost. Hence the sponsorship profile orients the edges of the complete graph. Since the average number of sponsored links is \((n-1)/2\), some agent sponsors \(d\ge h\) links. This agent cannot profitably delete all links she sponsors, so \(b_i(G_n)-b_i(G_n^{-d})\ge dc\).

By the supermodularity argument in Proposition~\ref{prop:supermodular-utility}, and by symmetry of the complete graph, the marginal loss from deleting an additional incident link is weakly decreasing as more incident links are deleted. Hence the average loss per deleted link is weakly decreasing in \(d\). Since \(d\ge h\), the inequality \(b_i(G_n)-b_i(G_n^{-d})\ge dc\) implies \(b_i(G_n)-b_i(G_n^{-h})\ge hc\).

Conversely, suppose \(b_i(G_n)-b_i(G_n^{-h})\ge hc\). Consider a balanced complete sponsorship profile, in which every edge is sponsored by exactly one endpoint and every agent sponsors either \(\lfloor (n-1)/2\rfloor\) or \(h\) links. Such a profile exists by orienting the complete graph as evenly as possible.

Since the graph is complete, no agent can gain by adding links. It remains only to rule out deletions. Let agent \(i\) sponsor \(d_i\le h\) links. Since the average loss per deleted link is weakly decreasing in the number of deleted links, the condition for \(h\) links implies \(b_i(G_n)-b_i(G_n^{-d_i})\ge d_i c\). Thus agent \(i\) does not gain by deleting all links she sponsors. By supermodularity and symmetry, if deleting all sponsored links is not profitable, then deleting any subset of sponsored links is not profitable. Therefore no agent has a profitable deviation, and the complete graph is an equilibrium network.
\end{proof}

\subsection{Implication for Katz--Bonacich Centrality}

For Katz--Bonacich centrality, the feasibility conditions can be written explicitly. Let \(m=n-k\) be the number of peripheral agents, let \(p\) be a peripheral agent, and let \(r=m-1=n-k-1\). For \(0<k\leq n-2\), define
\[
B_k=\frac{\delta k(1+\delta m)}{1-\delta(k-1)-\delta^2km}
\quad\text{and}\quad
A_k=\frac{\delta(k+r)+\delta^2r(k+1)}{1-\delta k-\delta^2r(k+1)}.
\]
Here \(B_k=b_p(G_k)\), while \(A_k=b_p(G_k^{+p})\), the benefit obtained by \(p\) after adding links to all other peripheral agents. Since deleting all links to the core isolates \(p\), we have \(b_p(G_k^{-p})=0\). Therefore a non-complete \(k\)-core equilibrium exists if and only if
\[
\frac{A_k-B_k}{n-k-1}\leq c\leq \frac{B_k}{k}.
\]

For \(k=0\), the deletion condition is vacuous. The empty network is an equilibrium if and only if \(b_p(G_0^{+p})\leq (n-1)c\), where
\(b_p(G_0^{+p})=\delta(n-1)(1+\delta)/(1-\delta^2(n-1))\).

For \(k=n\), Proposition~\ref{prop:complete-feasibility} gives the complete-graph condition. In the Katz--Bonacich case,
\[
b_i(G_n)-b_i(G_n^{-h})
=
\frac{\delta(n-1)}{1-\delta(n-1)}
-
b_i(G_n^{-h}),
\]
where \(G_n^{-h}\) is the complete graph with \(h=\lceil(n-1)/2\rceil\) links incident to \(i\) deleted. Thus the complete graph is an equilibrium network if and only if
\(b_i(G_n)-b_i(G_n^{-h})\ge hc\).

\section{Cost structure of Core-Periphery equilibria with Katz centrality}
\label{sec:convexcost}
We next consider a convex cost instead of the linear cost assumed so far. Let $G_{m,k}$ be a graph with $m$ nodes in the core and $k$ nodes in the periphery. The core nodes form a complete graph (a clique), while each periphery node is connected to the core nodes only. 
We would like to show how to construct equilibria that support this non-trivial core-periphery network.

\subsection{The candidate equilibrium profile $s^*$} 
Consider the core-periphery strategy profile $s^*=(s^*_z)_{z=1}^{m+k}$ where all periphery nodes use the same strategy, $i\in P= \{m+1,\dots,m+k\}$: each sponsors $m$ connections to all the core nodes: $s^*_i=\{1,\dots,m\}$. If $m$ is odd, each core node, $j\in  \{1,\dots,m\}$, sponsors $\frac{m-1}2$ edges to the other core nodes: $s^*_j=\{(j+1)\!\mod{m},\dots,(j+\frac{m-1}2)\!\!\mod{m}\}$. 
All core nodes sponsor distinct edges. 
The subgraph consisting of the $m$ core nodes has $\frac{m(m-1)}2$ edges and so is a complete subgraph. 
If $m$ is even, let all  $j\in\{2,\dots,m\}$ create the $m-1$ complete graph as in the previous case. In addition,
let every  $j\in\{2,\dots,m\}$ sponsor its connection to node $1$, who sponsors none. So, in case $m$ is even, $s^*_1=\emptyset$, while $\abs{s^*_j}=\frac m2$, $j\in\{2,\dots,m\}$ and the core nodes $C=\{1,\dots,m\}$ form a complete subgraph. In this case node one has no profitable deviations as it can not delete or add edges.
\subsection{Necessary and sufficient conditions for the supporting cost function}

\begin{theorem}\label{Thm_mk}
Let the benefit be of Katz-Bonacich type,
\begin{align}
b(A)=\delta A(I-\delta A)^{-1}\mathbf{1}=\left(\sum_{k=1}^\infty \delta^k A^k\right)
 \mathbf{1},\quad \delta<\frac1{n-1}
\end{align} 

For any $m\ge 1$ and any $k=n-m>m$,  the core-periphery strategy $s^*$ is a Nash equilibrium of the network formation game with payoffs 
\begin{align}
u_i(s)=b_i(A(s))-Q(\abs{s_i}),\quad \forall i\in\{1,\dots,n\}
\end{align}
for any increasing weakly convex cost function $Q\colon\{1,\dots, n\}\to\real_+$ that satisfies the following conditions:
\begin{align}
&Q(0)=0\\
&\label{cup}Q(t)\le \Upsilon(t;m,k,\delta)
\quad \forall t\in\{1,\dots,m\}\\
&\label{cd}Q(t)\ge \Lambda(t;m,k,\delta)
\quad \forall t\in\{m+1,\dots,m+k-1\},
\end{align}
where $\Upsilon$ and $\Lambda$ are increasing and convex in $t$ and can be computed from $m,k,\delta$ where $\delta$ is in the radius of convergence of the adjacency matrix of the graph $G_{m,k}$ and its perturbations where one periphery agent severs up to $m$ edges to the core agents or adds up to $k-1$ edges to other periphery agents. 
\end{theorem}
\begin{observation}
   \begin{enumerate} 
   \item If either one of the conditions $\eqref{cup}$ or $\eqref{cd}$ is violated, $s^*$ is not an equilibrium.
   \item The two bounds can be used to derive restrictions on parameters allowing the cost to be linear, as in our examples in Figure \ref{fig:moons}.
\end{enumerate}
\end{observation}
The proof relies on a series of auxiliary statements  showing that the binding deviation is always that of a periphery node. 
It is in the appendix.
\section{Conclusion}

We study a unilateral link sponsorship game in which agents trade off linear costs against benefits from a broad class of walk based centralities with positive and weakly decreasing weights. Our main result provides a sharp equilibrium prediction: every Nash equilibrium network is core periphery. The key force is structural complementarity. Linking to already well connected nodes creates many additional short walks at once, and since short walks receive higher weight, incentives concentrate links on a dense subset that serves as the core.

We also characterize efficiency. Despite the richness of equilibrium core sizes, the welfare maximizing network is always extreme: either empty or complete. This highlights a general tension between individually optimal stratification and socially optimal integration in environments where access benefits are cumulative.

Several extensions are natural. One direction is to relax homogeneous link costs or allow directed benefits, which may generate richer patterns of sponsorship while preserving the walk based logic. Another is to study increasing or convex costs more systematically and to quantify how the size of the equilibrium core varies with parameters across specific centrality specifications such as Katz and total communicability.

\bibliography{bib}
\appendix   
\section{Proof of the Injection Lemma \ref{lemma:deviation} and Proposition \ref{props:aux2}}
Let $G=(N,E)$ be a graph and let 
$i,j\in N$ be two distinct nodes.
Let $T_j\subseteq N_G(j)\setminus\{i\}$ be a subset of neighbors of $j$ that does not include agent $i$ such that $T_j\cap N_G(i)=\varnothing$. Assume that $N_G(j)\setminus\{i\}\subseteq T_j\cup \big(N_G(i)\big)$. That is $T_j$ is the subset of all agents that are neighbors of $j$ (not including $i$) that are not neighbors of $i$. Let $E_j=\{\{j,f\}|f\in T_j\}\}.$ 

Let $G'=(V,E')$ the graph that is obtained by adding the edges $E_i=\{\{i,f\}|f\in T'_j\}\}$ such that $T_j\subseteq T'_j$. 
Note that $N_{G'}(j)\setminus\{i\}\subseteq N_{G'}(i)\setminus\{j\}$.

We define a mapping \(F_i^j:\Psi^{E_j}_{G,j}\to\Psi^{E_i}_{G',i}\).
Let \(\gamma=(v_0,\ldots,v_T)\in\Psi^{E_j}_{G,j}\), where \(v_0=j\). Let \(\tau\) be the first index such that \(\{v_\tau,v_{\tau+1}\}\in E_j\). Thus \(\{v_\tau,v_{\tau+1}\}=\{j,k\}\) for some \(k\in T_j\). Let \(\sigma\) be the transposition of \(i\) and \(j\), that is, \(\sigma(i)=j\), \(\sigma(j)=i\), and \(\sigma(f)=f\) for every \(f\notin\{i,j\}\). We define \(\gamma'=F_i^j(\gamma)\) as follows.

We first define an auxiliary initial segment \(P_q\), for every \(q\leq \tau-1\). This is the image of the segment \((v_0,\ldots,v_q)\). If \(q=0\), set \(P_q=(i)\). If \(q\geq 1\) and \(v_1\neq i\), set \(P_q=(i,v_1,\ldots,v_q)\). Finally suppose that \(q\geq 1\) and \(v_1=i\). Let \(t\) be the first index in \(\{2,\ldots,q\}\) such that \(v_t\notin\{i,j\}\), with the convention that \(t=q+1\) if no such index exists. If \(t=q+1\), set \(P_q=(i,\sigma(v_1),\ldots,\sigma(v_q))\). If \(t\leq q\) and \(\{\sigma(v_{t-1}),v_t\}\in E\), set \(P_q=(i,\sigma(v_1),\ldots,\sigma(v_{t-1}),v_t,\ldots,v_q)\). If \(t\leq q\) and \(\{\sigma(v_{t-1}),v_t\}\notin E\), then necessarily \(\sigma(v_{t-1})=j\), and we delete this last switched \(j\); in this case set \(P_q=(i,\sigma(v_1),\ldots,\sigma(v_{t-2}),v_t,\ldots,v_q)\).

We now define \(\gamma'\) according to the orientation of the first edge from \(E_j\) traversed by \(\gamma\). 
First suppose that \((v_\tau,v_{\tau+1})=(j,k)\). Let \(a\) be the first index of the maximal block in \(\{i,j\}\) ending at \(v_\tau=j\); that is, \(v_a,\ldots,v_\tau\in\{i,j\}\), and either \(a=0\) or \(v_{a-1}\notin\{i,j\}\). If \(a=0\), set \(\gamma^-=(\sigma(v_0),\ldots,\sigma(v_\tau))\). If \(a>0\), then \(P_{a-1}\) ends at \(v_{a-1}\). If \(\{v_{a-1},\sigma(v_a)\}\in E\), set
\(\gamma^-=P_{a-1}\circ(\sigma(v_a),\ldots,\sigma(v_\tau))\).
If \(\{v_{a-1},\sigma(v_a)\}\notin E\), then necessarily \(\sigma(v_a)=j\), and we delete this first switched \(j\); in this case set
\(\gamma^-=P_{a-1}\circ(\sigma(v_{a+1}),\ldots,\sigma(v_\tau))\).
Finally, set \(\gamma'=\gamma^-\circ(k,v_{\tau+2},\ldots,v_T)\).

Now suppose that \((v_\tau,v_{\tau+1})=(k,j)\). We keep \(P_{\tau-1}\) unchanged and modify the continuation from \(k\) onward. If the walk stops at \(j\), set \(\gamma'=P_{\tau-1}\circ(k,i)\). If the walk continues as \((k,j,f,\ldots)\) with \(f\neq i\), set \(\gamma'=P_{\tau-1}\circ(k,i,f,\ldots)\). It remains to consider the case in which the walk continues as \((k,j,i,\ldots)\). Let \(s\) be the first index \(s\geq \tau+2\) such that \(v_s\notin\{i,j\}\). If no such \(s\) exists, set \(\gamma'=P_{\tau-1}\circ(k,\sigma(v_{\tau+1}),\ldots,\sigma(v_T))\). Otherwise, if \(\{\sigma(v_{s-1}),v_s\}\in E\), set \(\gamma'=P_{\tau-1}\circ(k,\sigma(v_{\tau+1}),\ldots,\sigma(v_{s-1}),v_s,\ldots,v_T)\). Otherwise, necessarily \(\sigma(v_{s-1})=j\); in this case we delete this last switched \(j\) and set \(\gamma'=P_{\tau-1}\circ(k,\sigma(v_{\tau+1}),\ldots,\sigma(v_{s-2}),v_s,\ldots,v_T)\).

This completes the construction of \(\gamma'=F_i^j(\gamma)\).
The deletion clauses are precisely those needed to keep the resulting sequence a valid walk in \(G'\).
Thus \(F_i^j(\gamma)\) is a walk in \(G'\) starting at \(i\). Moreover, by Observation~\ref{obs:first-Ei-edge} below, it crosses an edge in \(E_i\). Hence \(F_i^j(\gamma)\in \Psi^{E_i}_{G',i}\).

\begin{observation}\label{obs:first-Ei-edge}
Let \(\gamma\in\Psi^{E_j}_{G,j}\), and let \(e=\{j,k\}\in E_j\) be the first edge from \(E_j\) crossed by \(\gamma\). Then the first edge from \(E_i\) crossed by \(\gamma'=F_i^j(\gamma)\) is \(\{i,k\}\).
\end{observation}

\begin{proof}
Since \(\gamma\) is a walk in \(G\), it cannot cross any of the added edges in \(E_i\). Before the first crossing of \(e=\{j,k\}\), the walk \(\gamma\) does not cross any edge from \(E_j\). In the construction of \(F_i^j\), the part of \(\gamma'\) preceding the corresponding edge \(\{i,k\}\) is obtained from this initial part of \(\gamma\) only by switching \(i\) and \(j\), and possibly deleting a switched \(j\). Hence any earlier crossing of an edge from \(E_i\) in \(\gamma'\) would correspond to an earlier crossing of an edge from \(E_j\) in \(\gamma\), contradicting the choice of \(e\). Therefore the first edge from \(E_i\) crossed by \(\gamma'\) is \(\{i,k\}\).
\end{proof}
For any walk $\gamma=(v_0,\dots,v_T)$ recall that $\ell(\gamma)=T$ is the length of $\gamma$.
We next show two key properties of $F^j_i$.
\begin{lemma}\label{lem:bijection}
The map $F_i^j:\Psi^{E_j}_{G,j}\to\Psi^{E_i}_{G',i}$ is injective and satisfies $\ell(F_i^j(\gamma))\le \ell(\gamma)$ for all $\gamma\in\Psi^{E_j}_{G,j}$.
\end{lemma}

\begin{proof}
The inequality \(\ell(F_i^j(\gamma))\leq \ell(\gamma)\) is immediate from the construction. We prove injectivity. Let \(\gamma'=F_i^j(\gamma)=(w_0,\ldots,w_{\widetilde T})\). Let \(s\) be the first index such that \(\{w_s,w_{s+1}\}=\{i,h\}\) for some \(h\in T_j\). By construction, this edge is obtained from the first edge \(\{j,k\}\in E_j\) crossed by \(\gamma\), and hence \(h=k\). Thus \(\gamma'\) identifies the first edge \(\{j,k\}\in E_j\) crossed by \(\gamma\). Moreover, the orientation of this edge identifies the orientation of the crossing: if \((w_s,w_{s+1})=(i,k)\), then the original crossing was \((j,k)\), while if \((w_s,w_{s+1})=(k,i)\), then the original crossing was \((k,j)\).

Suppose first that \((w_s,w_{s+1})=(i,k)\). Then the original crossing was \((j,k)\). In this case \(\gamma'\) has the form \(\gamma^-\circ(k,w_{s+2},\ldots,w_{\widetilde T})\), where \(\gamma^-=(w_0,\ldots,w_s)\). Hence the suffix after \(k\) is recovered directly as \((w_{s+2},\ldots,w_{\widetilde T})\).

It remains to recover the prefix ending at the original vertex \(v_\tau=j\). Let \(r\leq s\) be the first index of the maximal terminal block of \(\gamma^-\) contained in \(\{i,j\}\). If \(r=0\), then this terminal block is all of \(\gamma^-\), and applying \(\sigma\) recovers \((v_0,\ldots,v_\tau)\).

Suppose \(r>0\). Then \(w_{r-1}\notin\{i,j\}\). Apply \(\sigma\) to the terminal block \((w_r,\ldots,w_s)\). If \(\{w_{r-1},\sigma(w_r)\}\in E\), then the resulting block is the original block \(v_a,\ldots,v_\tau\). Otherwise, the construction must have deleted the first switched \(j\). The deleted vertex was the image under \(\sigma\) of an original \(i\); after applying \(\sigma\), inserting this missing \(i\) at the beginning of the block recovers \(v_a,\ldots,v_\tau\). The preceding part \((w_0,\ldots,w_{r-1})\) is \(P_{a-1}\). By the definition of \(P_{a-1}\), this segment has a unique inverse: replace the initial \(i\) by \(j\), and, if an initial block in \(\{i,j\}\) was switched, apply \(\sigma\) back to that block, inserting the missing original \(i\) exactly when needed to obtain a valid walk in \(G\). Hence \((v_0,\ldots,v_{a-1})\) is uniquely recovered and so is
\[
\gamma=(v_0,\ldots,v_\tau,k,w_{s+2},\ldots,w_{\widetilde T}).
\]

Now suppose that \((w_s,w_{s+1})=(k,i)\). Then the original crossing was \((k,j)\). In this case \((w_0,\ldots,w_{s-1})=P_{\tau-1}\), and \(w_s=k\). By the definition of \(P_{\tau-1}\), the segment \((w_0,\ldots,w_{s-1})\) uniquely determines \((v_0,\ldots,v_{\tau-1})\): one replaces the initial \(i\) by \(j\), and, if an initial block in \(\{i,j\}\) was switched, applies \(\sigma\) back to that block, inserting the missing original \(i\) exactly when needed to obtain a valid walk in \(G\). Hence \((v_0,\ldots,v_{\tau-1})\) is recovered, and \(w_s=k\) gives \(v_\tau=k\).

It remains to recover the continuation after the original crossing \((k,j)\). If \(s+1=\widetilde T\), then the image stops at \(i\), so the original walk stopped at \(j\). If \(s+2\leq \widetilde T\) and \(w_{s+2}\notin\{i,j\}\), then the original continuation was \((j,w_{s+2},\ldots,w_{\widetilde T})\), and it is recovered directly.

Finally suppose that \(s+2\leq \widetilde T\) and \(w_{s+2}\in\{i,j\}\). Let \(r\geq s+1\) be the last index of the maximal block in \(\{i,j\}\) starting at \(w_{s+1}=i\). Apply \(\sigma\) to the block \((w_{s+1},\ldots,w_r)\). If \(r=\widetilde T\), this recovers the original continuation. If \(r<\widetilde T\) and \(\{\sigma(w_r),w_{r+1}\}\in E\), then the resulting block followed by \((w_{r+1},\ldots,w_{\widetilde T})\) recovers the original continuation. If \(r<\widetilde T\) and \(\{\sigma(w_r),w_{r+1}\}\notin E\), then the construction must have deleted the last switched \(j\). The deleted vertex was the image under \(\sigma\) of an original \(i\); after applying \(\sigma\), inserting this missing \(i\) at the end of the block recovers the original continuation.
\end{proof}
We next proceed to the proof of Lemma \ref{lemma:deviation}.
Recall that $G=(N,E)$ is a graph and $i\in N$. There are $k$ distinct agents $j_1, \dots, j_k$ (different from $i$) and associated edge sets $E_1, \dots, E_k$ satisfying four conditions:
\begin{enumerate}
    \item $E_l \subseteq E$ and $j_l \in e$ for every $e \in E_l$.
    \item $E_l \cap E_m = \varnothing$ for distinct $l, m$.
    \item  If $\{j_l, f\} \in E_l$, then $f \notin N_G(i) \cup \{i\}$.
    \item If $\{j_l, f\} \in E_l$ and $\{j_m, g\} \in E_m$ with $l \neq m$, then $f \neq g$.
\end{enumerate}
Let $G'$ be the graph obtained by adding the set of edges $E_i = \{ \{i,f\} : \exists l \text{ s.t. } \{j_l, f\} \in E_l \}$ to $G$. 

\injectivelem*
\begin{proof}
If $k=1$ let $F=F^{j_1}_i$. The injectivity follows from the definition of $F^{j_1}_i$ and Lemma \ref{lem:bijection} since $F^{j_1}_i$ is injective.

Next consider the case where $k\geq 2$. Note first that
$\Psi^{E_l}_{G,j_l}\cap\Psi^{E_f}_{G,j_f}=\varnothing$
for any two distinct indices $l,f\in\{1,\dots,k\}$.
This follows since every walk in $\Psi^{E_l}_{G,j_l}$ starts at $j_l$,
whereas every walk in $\Psi^{E_f}_{G,j_f}$ starts at $j_f$, and
$j_l\neq j_f$.

Thus we define $F:\Psi\to\Psi^{E_i}_{G',i}$ as follows. For
$\gamma\in\Psi^{E_l}_{G,j_l}$, set
\[
F(\gamma)=F^{j_l}_i(\gamma),
\]
where $F^{j_l}_i$ is the map constructed in Lemma~\ref{lem:bijection}
with respect to the graph $G'$. Since the sets
$\Psi^{E_l}_{G,j_l}$ are pairwise disjoint, $F$ is well defined.
Moreover, by Lemma~\ref{lem:bijection},
\[
\ell(\gamma)\geq \ell(F(\gamma))
\]
for every $\gamma\in\Psi$.

To see that \(F\) is injective, suppose that
\(F(\gamma)=F(\hat\gamma)\). Let
\(\gamma\in\Psi^{E_l}_{G,j_l}\) and
\(\hat\gamma\in\Psi^{E_m}_{G,j_m}\).

First suppose that \(l=m\). Then both walks belong to the same domain
\(\Psi^{E_l}_{G,j_l}\), and on this domain \(F\) coincides with
\(F_i^{j_l}\). Since \(F_i^{j_l}\) is injective by
Lemma~\ref{lem:bijection}, we have \(\gamma=\hat\gamma\).

It remains to consider the case \(l\neq m\). By
Observation~\ref{obs:first-Ei-edge}, the first edge from \(E_i\)
traversed by \(F(\gamma)\) is \(\{i,h\}\), where
\(\{j_l,h\}\in E_l\). Similarly, the first edge from \(E_i\)
traversed by \(F(\hat\gamma)\) is \(\{i,h'\}\), where
\(\{j_m,h'\}\in E_m\). By Condition~4 above, \(h\neq h'\). Hence the
first edge from \(E_i\) traversed by \(F(\gamma)\) is different from the
first edge from \(E_i\) traversed by \(F(\hat\gamma)\). Therefore
\(F(\gamma)\neq F(\hat\gamma)\), contradicting our supposition.
Thus  \(F\) is
injective.

To see that $F$ is not onto, choose an edge $\{j_1,f\}\in E_1$. Then
$\{i,f\}\in E_i$, and hence the walk $\gamma'=(i,f,i)$ belongs to
$\Psi^{E_i}_{G',i}$. We claim that $\gamma'\notin F(\Psi)$.

Suppose otherwise that $\gamma'=F(\gamma)$. Since $\gamma'$ uses the
unique edge $\{i,f\}$ from $E_i$, Observation~\ref{obs:first-Ei-edge}
implies that the first special edge used by $\gamma$ must be
$\{j_1,f\}$. In the construction of $F_i^{j_1}$, after the edge
$\{i,f\}$ is created, the image walk continues from $f$ only along
edges of the original graph $G$. Thus, in order for the image to move
next from $f$ to $i$, we would need $\{f,i\}\in E$. This contradicts
Condition~3, which implies $f\notin N_G(i)$. Therefore
$\gamma'\notin F(\Psi)$, and $F$ is not surjective.
\end{proof}

We next turn to prove Proposition \ref{props:aux2}.
Recall that $i,j$ are such  that $N_G(i)\setminus\{j\}\not\subseteq N_G(j)\setminus\{i\}$. That is some neighbor of $j$ is not a neighbor of $i$. Let $T_j=\big(N_G(j)\setminus N_G(i)\big)\setminus\{i\}$ be the set of such agents. Let $E_j=\{\{j,k\}:k\in T_j\}\subseteq E$ and let $E_i=\{\{i,k\}:j\in T_j\}$. By assumption $E_i\cap E=\varnothing$. Let $G'=(N,E')$ be network that is obtained by adding the  edges $E_i$ to $E$. Let $e'=\{i,k\}\in E_i$ and let $e=\{j,k\}\in E$. 
\auxprop*
\begin{proof}[\textbf{Proof of Proposition \ref{props:aux2}}]
We define a  mapping $H:\Psi^{E_j}_{G}\to\Psi^{E_i}_{G'}$ that is one to one and satisfies $\ell(\gamma)\leq \ell(F(\gamma))$ for every $\gamma\in \Psi^{E_j}_{G}$. To do that, for every node $k\in N$ let $\Psi^{E_j}_{G,k}$ be the walks in $\Psi^{E_j}_{G}$ starting at node $k$. Define  $\Psi^{E_i}_{G',k}$ similarly. 

Let $k\neq i,j$. We first define $F_k:\Psi^{E_j}_{G,k}\to\Psi^{E_i}_{G',k}$. Let $\gamma=(v_0,\ldots,v_T)\in\Psi^{E_j}_{G,k}$ and let $0\leq\tau\leq T-1$ be the first time that $\{v_\tau,v_{\tau+1}\}\in E_j$. 
Thus $\{v_\tau,v_{\tau+1}\}=\{j,f\}$ for some $f 
\in T_j$.  

Define $\gamma'=(v'_0,\ldots,v'_{T'})\in\Psi^{E_i}_{G',k}$ as follows, distinguishing cases according to the \emph{directed} traversal of the first edge in $E_j$.

\smallskip
\noindent\textbf{Case 1:} $(v_\tau,v_{\tau+1})=(f,j)$ for some $f\in T_j$.
Set $T'=T$, put $v'_{\tau+1}\defeq i$, and $v'_t\defeq v_t$ for all $t\neq \tau+1$.
(Then $\{v'_\tau,v'_{\tau+1}\}=\{f,i\}\in E_i$.)

\smallskip
\noindent\textbf{Case 2:} $(v_\tau,v_{\tau+1})=(j,f)$ for some $f\in T_j$.

\smallskip
\noindent\emph{Case 2a:} $v_{\tau-1}\neq i$.
Set $T'=T$, put $v'_\tau\defeq i$, and $v'_t\defeq v_t$ for all $t\neq \tau$.

\smallskip
\noindent\emph{Case 2b:} $v_{\tau-1}=i$.
Write $v_{\tau-2}=\ell$ (so the prefix is $\ell\to i\to j\to f$). If $\{\ell,j\}\in E$, set $T'=T$ and define
\[
(v'_{\tau-2},v'_{\tau-1},v'_{\tau},v'_{\tau+1})\defeq (\ell,j,i,f),
\]
with $v'_t\defeq v_t$ for all other $t$.
If $\{\ell,j\}\notin E$, set $T'=T-1$ and delete the intermediate visit to $j$ by defining
$v'_t\defeq v_t\ \ (0\le t\le \tau-1),\
v'_{t-1}\defeq v_t \ (\tau+1\le t\le T).$
(So the segment $i\to j\to f$ is replaced by $i\to f$, which is valid since $\{i,f\}\in E_i$.)

Note that by construction $\ell(F_k(\gamma))\leq \ell(\gamma)$. To see injectivity, fix $\gamma\in\Psi^{E_j}_{G,k}$ and let $\gamma'=F_k(\gamma)$. Since $\gamma'\in\Psi^{E_i}_{G',k}$, there exists a (unique) first time
$
t^\ast\defeq \min\{t:\{v'_t,v'_{t+1}\}\in E_i\}.$

We claim that $F_k$ (for $k\neq i,j$).
Fix $\gamma\in\Psi^{E_j}_{G,k}$ and let $\gamma'=F_k(\gamma)=(v'_0,\dots,v'_{T'})$. Since $E_i\cap E=\varnothing$, every traversal of an
edge in $E_i$ in $\gamma'$ is necessarily created by the replacement rule. Let
$t^\ast\defeq \min\{t:\{v'_t,v'_{t+1}\}\in E_i\}.$
Then $t^\ast$ is well defined and uniquely determines the location of the first replacement.

Write $\{v'_{t^\ast},v'_{t^\ast+1}\}=\{i,f\}$ with $f\in T_j$ and distinguish cases.

If $(v'_{t^\ast},v'_{t^\ast+1})=(f,i)$, then necessarily the original first $E_j$-traversal was $(f,j)$, so we recover $\gamma$ by setting
$v_{\tau+1}=j$ and $v_t=v'_t$ for all other $t$ (thus $T=T'$ and $\tau=t^\ast$).

If $(v'_{t^\ast},v'_{t^\ast+1})=(i,f)$, then the original first $E_j$-traversal was $(j,f)$, and we look at the predecessor of $i$.
Let $\ell$ denote the vertex preceding this $i$ in $\gamma'$ (if any), i.e.\ $\ell\defeq v'_{t^\ast-1}$ when $t^\ast\ge 1$.
\begin{itemize}
\item If $t^\ast=0$ or $\ell\neq j$, then we are in the ``direct'' case $v_{\tau}=j$ with $v_{\tau-1}\neq i$, and we invert by setting
$v_{\tau}=j$ and $v_t=v'_t$ for all other $t$ (so $T=T'$ and $\tau=t^\ast$).

\item If $\ell=j$, then the replacement must have been applied to a local pattern $\ell\to i\to j\to f$ in $\gamma$.
Now check whether $t^\ast\ge 2$ and whether the vertex $v'_{t^\ast-2}$ is adjacent to $j$ in $G$.
Write $v'_{t^\ast-2}=r$ (when $t^\ast\ge 2$).
If $\{r,j\}\in E$, then the forward map used the ``swap'' rule and did not change length; hence $T=T'$ and we invert locally by replacing
the triple $(r,j,i,f)$ in $\gamma'$ with $(r,i,j,f)$, leaving all other vertices unchanged.

If $\{r,j\}\notin E$, then the forward map used the ``contraction'' rule and deleted the intermediate visit to $j$; hence $T=T'+1$ and
we invert by replacing the triple $(r,i,f)$ in $\gamma'$ with the quadruple $(r,i,j,f)$, leaving all other vertices unchanged and shifting
subsequent indices by $+1$.
\end{itemize}
In all cases the preimage $\gamma$ is uniquely determined by $\gamma'$ and the graph adjacencies, so $F_k$ is injective.

We next consider the case of $k=j$. In this case, it follows from the construction Lemma \ref{lemma:deviation} that there exists an injection $F_j:\Psi^{E_j}_{G,j}\to \Psi^{E_i}_{G',i}$. Next we define a mapping $F_i:\Psi^{E_j}_{G,i}\to \Psi^{E_j}_{G',j}$. Let $\gamma=(v_0,\ldots,v_T)\in \Psi^{E_j}_{G,i}$. We define $\gamma'=(v'_0,\ldots,v'_T)\in F_i(\gamma)$ Let $\tau$ be the minimal $t\leq T-1$ such that $\{v_\tau,v_{\tau+1}\}=\{j,f\}\in E_j$. If $v_\tau=j$ define 
$\gamma'$ as follows: $\gamma'=(v_\tau,v_{\tau-1}\ldots,v_0,v_{\tau+1},v_{\tau+2}\ldots,v_T)$. Thus $\gamma'$ starts at $v_\tau=j$ and then moves \emph{backwards} untill it reaches $v_0=i$ then continues at $v_{\tau+1}$ and then follows $\gamma$. Note that by construction $(i,f)=(v_0,v_{\tau+1})$ and $\{i,f\}\in E_i$ so $\gamma'\in\Psi^{E_i}_{G',j}$. 

Similarly if $(v_\tau,v_{\tau+1})=(f,j)$ let $\gamma'=(v_{\tau+1},v_{\tau}\ldots,v_0,v_{\tau+1},\ldots,v_T)$. Again by construction $(i,f)=(v_0,v_{\tau+1})\in E_i$. This completes the definition of $F_i$. It is easy to see that $legth(\gamma)=\ell(F_i(\gamma).$ In this case, injectivity follows stright for the definition. To complete the proof we need, using the maps $\{F_k\}_{k\in N},$ to define the mapping $H:\Psi^{E_j}_{G}\to\Psi^{E_i}_{G'}$. 

We now complete the proof. Using the preceding construction, for every node $k\in N$ we have defined an injective map
\[
F_k:\Psi^{E_j}_{G,k}\to \Psi^{E_i}_{G',k}
\quad\text{such that}\quad
\ell(\gamma)\ge \ell(F_k(\gamma))\ \ \text{for all }\gamma\in\Psi^{E_j}_{G,k}.
\]
(For $k\neq i,j$ this is the map constructed above. For $k=j$ we use the injection $F_j:\Psi^{E_j}_{G,j}\to\Psi^{E_i}_{G',i}$, and for $k=i$ we use the map $F_i:\Psi^{E_j}_{G,i}\to\Psi^{E_i}_{G',j}$ defined above.)

Define $H:\Psi^{E_j}_{G}\to\Psi^{E_i}_{G'}$ by
\[
H(\gamma)\defeq F_{v_0}(\gamma)\qquad\text{for }\gamma=(v_0,\ldots,v_T)\in\Psi^{E_j}_{G}.
\]
Since the sets $\{\Psi^{E_j}_{G,k}\}_{k\in N}$ are disjoint and each $F_k$ is injective, $H$ is injective. Moreover, by construction,
\[
\ell(\gamma)\ge \ell(H(\gamma))\qquad\text{for all }\gamma\in\Psi^{E_j}_{G}.
\]
Because $a_1\ge a_2\ge\cdots>0$, it follows that
\[
b^{E_j}_{G}
=\sum_{\gamma\in\Psi^{E_j}_{G}}a_{\ell(\gamma)}
\le \sum_{\gamma\in\Psi^{E_j}_{G}}a_{\ell(H(\gamma))}
<\sum_{\eta\in\Psi^{E_i}_{G'}}a_{\ell(\eta)}
=b^{E_i}_{G'}.
\]
The reason for the strict inequality follows from the fact that, as we showed in Lemma \ref{lem:bijection} there exists a walk $\gamma'\in\Psi^{E_i}_{G'}$ such that $\gamma'\neq H(\gamma)$ for every $\gamma\in\Psi^{E_j}_{G}$.

Hence adding the edges $E_i$ (going from $G$ to $G'$) increases the walk term by at least $b^{E_i}_{G'}-b^{E_j}_{G}\ge 0$, while the edge term decreases by exactly $-|E_i|=-|E_j|$. Therefore
\[
W_{G'}-W_G
=\bigl(\sum_{\gamma\text{ walk in }G'}a_{\ell(\gamma)}-\sum_{\gamma\text{ walk in }G}a_{\ell(\gamma)}\bigr)-(|E'|-|E|)
\ge b^{E_i}_{G'}-|E_i|
\ge b^{E_j}_{G}-|E_j|\geq 0,
\]
as desired. This completes the proof of the proposition.
\end{proof}

\subsection{Proof of Theorem \ref{Thm_mk}}
\subsubsection{Possible deviations from the core-periphery strategy $s^*$.}
None of the nodes from the core, $\{1,\dots,m\}$, can add more edges, as  they are connected to all other nodes  under $s^*$.
Any core node can deviate by eliminating $q\le\lfloor{\frac{m}2}\rfloor$ sponsored edges. However, this deviation never imposes any binding restrictions on costs, as such deviation will never be profitable as long as neither of the
periphery nodes is interested in deleting any of their $q\le m$ edges to the core. 

Indeed, a core node $c$ that dropped $q$ edges to its core neighbors and is left with $m-1-q$ edges to the core, has a higher benefit  than a periphery node $p$ with the same set of core neighbors. This follows from Lemma \ref{lemma:deviation}: the set of neighbors of the core node $c$ includes all the neighbors  of $p$ and, in addition, the set of all periphery nodes $P$.

Thus it is sufficient to only rule out deviations of the periphery nodes.   (1) A deviator can possibly delete $q\le m$ edges  to the core nodes and (2) add $q\le k-1$ edges to the other periphery nodes. 
Both deviations are perturbations of the adjacency matrix:
$A_i^q=\A-E_i^q,\quad i=1,2$,
so that%
\begin{align*}
E_1^q=we_p^{\top}+w^{\top}e_p,\quad E_2^q=-ve_p^{\top}-v^{\top}e_p,
\end{align*}
 where  $p\in P=\{m+1,m+k\}$ denotes the deviating periphery node, $e_i\in\real^n$ denotes a standard basis (column) vector, 
 $w=\sum_{i\in C'\subset C}e_i$ is the vector of indicators of the $q$ disconnected core elements and
$v=\sum_{i\in P'\subset P}e_i$ is the vector of indicators of the $q$ newly connected periphery elements. 
Both $E_i,i=1,2$ are rank 2 matrices. The change  in benefits $(\Delta b(A^i_q))_p$ of the node $p$ will yield the corresponding bounds on the cost function. We  calculate this change for any number of modified edges ($q$). Clearly, the value is the same for any deviating periphery node.
We start by calculating the benefits under the candidate equilibrium profile $s^*$.

\subsubsection{Benefits under $s^*$.}
For the specification of Katz-Bonacich centrality given in the theorem, the benefit vector  can be presented as follows
\[
b(A) = y - \mathbf{1},
\quad
y = \R\mathbf{1},\quad 
\R = (M)^{-1},\quad M=I - \delta A
\]

Given strategy \(s^*\), the adjacency matrix \(A(s^*)=\A\) is
\[
\A =
\begin{pmatrix}
J_m - I_m & J_{m\times k} \\
J_{k\times m} & 0
\end{pmatrix},
\]
where \(J_{m\times k}\) is the \(m\times k\) all-ones matrix,
\(J_m=J_{m\times m}\), and \(I_m\) is the \(m\times m\) identity.
Then 
\[
M=I - \delta \A
=
\begin{pmatrix}
(1+\delta) I_m - \delta J_m & -\delta J_{m\times k} \\
-\delta J_{k\times m} & I_k
\end{pmatrix}.
\]
The resolvent \(\R=(I-\delta \A)^{-1}\) has the  block form
\[
\R \;=\;
\begin{pmatrix}
\dfrac{1}{1+\delta}\, I_m \;+\; \dfrac{\delta(1+\delta k)}{D(1+\delta)}\, J_m
&
\dfrac{\delta}{D}\, J_{m\times k}
\\[10pt]
\dfrac{\delta}{D}\, J_{k\times m}
&
I_k \;+\; \dfrac{\delta^2 m}{D}\, J_k
\end{pmatrix},
\quad
D = 1 - \delta(m-1) - \delta^2 k m.
\]
Each component $i\in C$ from the first $m$ components of $y$ is
\[
y_C =\frac{1}{1+\delta}+\frac{m\delta(1+\delta k)}{D(1+\delta)}+\frac{k\delta}{D}= \frac{1+\delta k}{D}.
\]
Each component $j\in P$ of the last $k$ components of $y$ are
\[
y_P =\frac{m\delta}{D}+1+\frac{k\delta^2 m}{D}
= \frac{1+\delta}{D}.
\]
The utility of the periphery agent $j\in P$ under $s^*$ is 
$u_P(s^*)=y_P-1-Q({m})$.
If $m$ is odd, the utility of the core $i\in C$ agent  is 
$u_C(s^*)=y_C-1-Q(\frac{m-1}2)$ and if $m$ is even
\begin{align*}
u_1(s^*)=y_C-1,\quad u_i(s^*)=y_C-1-Q(\frac{m}2),\quad i=2,\dots,m
\end{align*}
A core agent has a higher utility under $s^*$: 
$u_C(s^*)-u_P(s^*)\ge y_C-y_P\ge\frac{\delta}D(k-1)>0$.

\subsubsection{The deviations benefits}
The change in benefit before and after deviation $i=1,2$ for node $p\in P$ is the $p$-th element of the vector 
\begin{align}
\Delta b(A_i^q)=(\R-\R_i^q)\mathbf{1}, \quad \R_i^q=(M_i^q)^{-1},\quad M_i^q=I-\delta A_i^q=I - \delta(\A-E_i^q)
\end{align}
Each deviation (for any number of modified links $q$) is a rank-2 perturbation of the adjacency matrix:
\begin{align}
    &E_i^q\;=\;U_i^q (V_i^q)^{\top},\\
    &U_1^q=\bigl[\,w\ \ e_p\,\bigr],\quad
   V_1^q=\bigl[\,e_p\ \ w\,\bigr],\quad
   U_2^q\;=\; \bigl[\,-v\ \ -e_p\,\bigr],\quad
V_2^q=\bigl[\,e_p\ \ v\,\bigr]\\
&w=\sum_{i\in C'\subset C}e_i ,\quad
v=\sum_{i\in P'\subset P}e_i,\quad|C'|=q,\quad|P'|=q
\end{align}
Using Woodbury low rank update,%
\footnote{If $M\in\mathbb{R}^{n\times n}$ is invertible and $U,V\in\mathbb{R}^{n\times z}$, then
\[
(M+U V^\top)^{-1}
\;=\;
M^{-1} \;-\; M^{-1} U\bigl(I_{z} + V^\top M^{-1} U\bigr)^{-1} V^\top M^{-1}.
\]}
\[
\R_i^q
\;=\;
\R
\;-\;
\R\,(\delta U_i^q)\,
\Bigl(I_{2} + (V_i^q)^\top \R\,(\delta U_i^q)\Bigr)^{-1}(V_i^q)^\top \R
\]

and the identity $y = \R\mathbf{1}$, we get
\begin{align}
\Delta b(A_i^q)=\delta\R U_i^q\,
(H_i^q)^{-1}(V_i^q)^\top y,\quad H_i^q=I_{2} + \delta (V_i^q)^\top \R U_i^q
\end{align}
The determinant $h_i(q)$ of the $2\times2$ matrix $H_i^q$ is strictly positive for the assumed range of $\delta$ for both deviations $i$.

To calculate the loss from dropping the edges, we compute
\[
V_1^\top y
=
\begin{pmatrix}
e_p^\top y\\[2pt]
w^\top y
\end{pmatrix}
=
\begin{pmatrix}
y_P\\[2pt]
q\,y_C
\end{pmatrix},\quad
e_p^\top \R U_1
=
\Bigl[\ \tfrac{\delta q}{D}\,,\ 1+\tfrac{\delta^2 m}{D}\ \Bigr].
\]
\[
(H_1^q)^{-1}
=
\bigl(I_2+\delta V_1^\top \R U_1\bigr)^{-1}
=
\frac{1}{h_1(q)}
\begin{pmatrix}
\displaystyle 1+\frac{\delta^2 q}{D}
&
\displaystyle -\Bigl(\delta + \frac{\delta^3 m}{D}\Bigr)
\\[10pt]
\displaystyle -\Bigl(\frac{\delta}{1+\delta}\,q
+\frac{\delta^2(1+\delta k)}{D(1+\delta)}\,q^2\Bigr)
&
\displaystyle 1+\frac{\delta^2 q}{D}
\end{pmatrix},
\]
where 
\[
h_1(q)
=
\Bigl(1+\frac{\delta^2 q}{D}\Bigr)^{\!2}
\;-\;
\frac{\delta}{1+\delta}\,
\Bigl(1+\frac{\delta^2 m}{D}\Bigr)\,
\Bigl(\delta q+\frac{\delta^2(1+\delta k)}{D}\,q^2\Bigr)>0.
\]
The change in benefit before and after deleting $q$ edges is
\begin{align}
&(\Delta b(A_1^q))_p
=
\frac{\delta\,q\;\bigl(\alpha_1+\beta_1\,q\bigr)}{h_1(q)}>0,\\
&\alpha_1
=
y_C
\;+\;
\frac{\delta^2 m\,\bigl(2+\delta(2k-1)\bigr)}{D^{2}}>0,\quad
\beta_1
=
-\,\frac{\delta^2\bigl(1+\delta(k-1)\bigr)}{D^2}<0,\\
&h_1(q)
=
1
+
a_1 q
+
a_2 q^2,\\
&a_1=\frac{\delta^2\bigl(1+\delta(m+1)+\delta^2 m (k-1)\bigr)}
{D\,(1+\delta)}>0,\quad
a_2=-\,\frac{\delta^3\bigl(1+\delta(k-1)\bigr)}
{D\,(1+\delta)}<0.
\end{align}
Similar derivation for the case of change in benefit before and after adding edges yields
\begin{align*}
&(\Delta b(A_2^q))_p
\;=\;
\frac{\delta\,q\bigl(\alpha_2+\beta_2 q\bigr)}{h_2(q)}
<0\\
&\alpha_2
\;=\;
-\left[\,y_P\,(D+\delta^2 m) + \frac{\delta^2 m}{D}\,\right]<0,
\qquad
\beta_2 \;=\; -\,\frac{\delta^3 m}{D}<0
\\
&h_2(q)
\;=\;
1
-\gamma_1\,q
-\gamma_2\,q^2,\\
&\gamma_1=\Bigl[\delta^2 + \frac{\delta^3 m}{D}(2+\delta)\Bigr]>0,\quad
\gamma_2=\frac{\delta^4 m}{D}>0
\end{align*}
\subsection{Average loss of dropping edges is decreasing}
Define a real-valued function on reals: $$g_1(q)=\frac{\delta\,(\alpha_1+\beta_1 q)}{1+a_1 q+a_2 q^2}=\frac{(\Delta b(A_1^q))_p}q.$$
Then $g_1(q)$ is decreasing in $q$ and is strictly positive for all $0\le q\le m$.
\begin{proof}
The sign of $g_1'(q)$ is the sign of 
\[
N(q):=\beta_1 h_1(q)-(\alpha_1+\beta_1 q)\,h_1'(q)= \beta_1 - \alpha_1 a_1 - 2\alpha_1 a_2\,q - \beta_1 a_2\,q^2.
\]

The quadratic polynomial \(N(q)=c_0+c_1 q+c_2 q^2\) has coefficients
\[
c_0=\beta_1-\alpha_1 a_1<0,\qquad
c_1=-2\alpha_1 a_2>0,\qquad
c_2=-\beta_1 a_2<0,
\]
because \(\alpha_1>0\), \(\beta_1<0\), \(a_1>0\), \(a_2<0\).
Hence \(N(q)\) is a downward-opening quadratic (\(c_2<0\))
whose maximum is at
\[
q^\star \;=\; -\frac{c_1}{2c_2}
\;=\; -\,\frac{-2\alpha_1 a_2}{2(-\beta_1 a_2)}
\;=\; \frac{\alpha_1}{\beta_1}\;<\;0,
\]
since \(\alpha_1>0\) and \(\beta_1<0\).
Therefore, on the feasible domain \(q\ge 0\), the function \(N(q)\) lies
strictly to the right of its (unique) maximum and is thus strictly decreasing.
Because \(N(0)=\beta_1-\alpha_1 a_1<0\), it follows that
$
N(q)<0$  for all  $q\ge 0$.
Thus, $g_1$ is decreasing.

$\Delta b(A_1^q))_p$ evaluated at $q=m$ is positive as it is the benefit of the periphery node under $s^*$, which equals $y_P-1>0$.
\end{proof}
\subsubsection{Average benefit from adding edges is increasing}
Let $g_2$ be defined on reals:
$$ g_2(q)=\frac{-(\alpha_2+\beta_2 q)}{h_2(q)}=-\frac{(\Delta b(A_2^q))_p}q.$$
Then $g_2(q)$ is strictly increasing in $q$.
\begin{proof}The derivative of $g_2$ is positive, because the numerator of its derivative is 
$$-\beta_2  h_2(q)+(\alpha_2+\beta_2 q)h'_2(q)$$ which is positive because
$h_2(q)>0$, $\beta_2<0$,
$h'_2(q)=\gamma_1-2\gamma_2<0$, $\alpha_2+\beta_2 q<0$.
\end{proof}
\subsubsection{Constructing the bounds on the supporting cost function.}
Cost of the first $m$ edges should not be too high in order to prevent any periphery agent from dropping any of its $m$ edges to the core. The loss of dropping edges is increasing in the number of edges dropped, $q$, and is concave in $q$, because the average loss is decreasing. For any $t\in \{1,\dots,m\}$ the benefit from having $t$ edges for a periphery agent is the same as the difference between the equilibrium payoff (having $m$ links to the core), which is $y_P-1$, and the change in the benefit due to dropping $q=m-t$ edges. This difference is increasing in $t$ and is  convex in $t$. So the first requirement is
\begin{align}\label{Cless}
Q(t)\le \Upsilon(t;m,k,\delta):= (y_P-1)-(\Delta b(A_1^q))_p|_{q=m-t},\quad t=\{1,\dots,m\}
\end{align}
To prevent the second type of deviation, the cost has to be above the benefit of adding any amount of edges up to $k-1$, which equals $-(\Delta b(A_2^q))_p|_{q=t-m}$. It is increasing and convex in the number of added edges. So, to prevent any addition of an edge by a peripheral node it is sufficient to require that the cost is high enough for $t\in\{m+1,\dots,m+k-1\}$:
\begin{align}\label{Cmore}
Q(t)\ge \Lambda(t;m,k,\delta):=-(\Delta b(A_2^q))_p|_{q=t-m},\quad t\in\{m+1,\dots,m+k-1\}
\end{align}

Both bounds can be written explicitly in terms of $\delta$, $t$, $m$, and $k$.

This completes the proof of Theorem \ref{Thm_mk}.

\end{document}